\documentclass[letterpaper,10pt,onecolumn,conference]{IEEEtran}

\IEEEoverridecommandlockouts                              % This command is only needed if
\title{\LARGE \bf
A Phase-space Formulation of the Belavkin-Kushner-Stratonovich
Filtering Equation for Nonlinear Quantum Stochastic Systems$^*$
}

\author{Igor G. Vladimirov$^{\dagger}$%: \today, \currenttime%\\
\thanks{$^*$This work is supported by the Australian Research Council.}
\thanks{$^\dagger$UNSW Canberra, Australia.
{\tt igor.g.vladimirov@gmail.com}.}
}

\usepackage{amsmath}
\usepackage{amssymb}
\usepackage{amsfonts}
\usepackage{graphicx}
\usepackage{color}
\usepackage{times}
\usepackage{mathptmx} % assumes new font selection scheme installed
\usepackage{datetime}

\usepackage{datetime}
\usepackage{times} % assumes new font selection scheme installed
\usepackage{graphicx} % for pdf, bitmapped graphics files
\usepackage{amsmath} % assumes amsmath package installed
\usepackage{amssymb}  % assumes amsmath package installed
\newtheorem{thm}{Theorem}

\def\fa{\mathfrak{a}}
\def\fb{\mathfrak{b}}

\def\<{\leqslant}           % nice less than or equal to sign
\def\>{\geqslant}           % nice larger than or equal to sign
\def\div{\mathrm{div}}         % divergence
\def\d{\partial}

\def\Re{\mathrm{Re}}   % real part
\def\Im{\mathrm{Im}}   % imaginary part

\def\cH{\mathcal{H}}   % Hardy space
\def\mR{\mathbb{R}}    % real line
\def\mC{\mathbb{C}}    % complex plane

\def\Tr{\mathrm{Tr}}       % matrix trace
\def\rT{\mathrm{T}}        % matrix transpose
\def\bE{\mathbf{E}}    % expectation

\def\[[[{[\![\![}
\def\]]]{]\!]\!]}

\def\bra{\langle}
\def\ket{\rangle}

\def\Bra{\left\langle}
\def\Ket{\right\rangle}

\def\re{\mathrm{e}}        % number e
\def\rd{\mathrm{d}}        % differential

\def\cL{\mathcal{L}}

\def\x{\times}
\def\ox{\otimes}

\def\fF{\mathfrak{F}}
\def\fG{\mathfrak{G}}

\def\cZ{\mathcal{Z}}
\def\fZ{\mathfrak{Z}}

\def\sM{\mathsf{M}}
\def\sK{\mathsf{K}}

\def\cF{\mathcal{F}}
\def\cW{\mathcal{W}}

\def\cC{\mathcal{C}}

\def\cI{\mathcal{I}}

\def\cA{\mathcal{A}}
\def\cB{\mathcal{B}}

\def\cS{\mathcal{S}}

\def\mS{\mathbb{S}}

\def\Ups{\Upsilon}
\def\ups{\upsilon}

\def\diag{\mathop{\mathrm{diag}}}    % diagonal matrix
\begin{document}
\maketitle
\thispagestyle{empty}
\pagestyle{plain}

%\subjclass[2000]{Primary }
%    For articles to be published after 1 January 2010, you may use
%    the following version:
%\subjclass[2010]{Primary:
%81S22, % Open systems, reduced dynamics, master equations, decoherence
%81S25, % Quantum stochastic calculus
%81S30, % Phase-space methods including Wigner distributions, etc.
%81P16, % Quantum state spaces, operational and probabilistic concepts
%81S05; %Canonical quantization, commutation relations and statistics
%secondary:
%81Q15, % Perturbation theories for operators and differential equations
%35Q40, %PDEs in connection with quantum mechanics
%37M25. % Computational methods for ergodic theory (approximation of invariant measures, computation of Lyapunov exponents, entropy
%}

%\keywords{
%Open quantum systems,
%canonical commutation relations,
%quantum stochastic differential equations,
%Weyl quantization,
%Wigner-Moyal approach,
%quasi-characteric functions,
%quasi-probability density functions,
%integro-differential equations,
%Renyi entropy.
%}
%
%\date{ \today, \currenttime}

\begin{abstract}
This paper is concerned with a filtering problem for a class of nonlinear quantum stochastic systems with multichannel nondemolition measurements. The system-observation dynamics are governed by a Markovian  Hudson-Parthasarathy  quantum stochastic differential equation driven by quantum Wiener processes of bosonic fields in vacuum state. The Hamiltonian and system-field coupling operators, as functions of the system variables,   are represented in a Weyl quantization form.  Using the Wigner-Moyal phase-space framework, we obtain a stochastic integro-differential equation for the posterior quasi-characteristic function (QCF) of the system conditioned on the measurements. This equation is a spatial Fourier domain representation of the Belavkin-Kushner-Stratonovich stochastic master equation driven by the innovation process associated with the measurements. We also discuss a more specific form of the posterior QCF dynamics in the case of linear system-field coupling and outline a Gaussian approximation of the posterior quantum state.
\end{abstract}

\maketitle

\thispagestyle{empty}

%%%%%%%%%%%%%%%%%%%%%%%%%%%%%%%%%%%%%%%%%%%%%%%%%%%%%%%%%%%%%%%%%%%%%%%%%%%%%%%%%%%%%%%%%%%%%%%%%%%

%\pagestyle{plain}

%%%%%%%%%%%%%%%%%%%%%%%%%%%%%%%%%%%%%%%%%%%%%%%%%%%%%%%%%%%%%%%%%%%%%%%%%%%%%%%%%%%%%%%%%%%%%%%%%%%
%%%%%%%%%%%%%%%%%%%%%%%%%%%%%%%%%%%%%%%%%%%%%%%%%%%%%%%%%%%%%%%%%%%%%%%%%%%%%%%%%%%%%%%%%%%%%%%%%%%

%%%%%%%%%%%%%%%%%%%%%%%%%%%%%%%%%%%%%%%%%%%%%%%%%%%%%%%%%%%%%%%%%%%%%%%%%%%%%%%%%%%%%%%%%%%%%%%%%%%
\section{INTRODUCTION}\label{sec:intro}
%%%%%%%%%%%%%%%%%%%%%%%%%%%%%%%%%%%%%%%%%%%%%%%%%%%%%%%%%%%%%%%%%%%%%%%%%%%%%%%%%%%%%%%%%%%%%%%%%%%

Estimation of the current state of a dynamical system, based on the past history of a statistically dependent random process, is a central problem in the stochastic filtering theory which dates back to the works of Kolmogorov and Wiener of the 1940s \cite{K_1941,W_1949}. The performance of a state estimator is often described in terms of mean square values of the estimation errors which have  to be minimized.
The filtering problems arise naturally when the system state (whose knowledge, precise or approximate, is required for feedback control) is not accessible to direct measurement. Such measurements are particularly problematic in regard to physical systems on the atomic scales, whose evolution is described in terms of operator-valued variables and obeys the laws of quantum mechanics which prohibit simultaneous projective measurements  of noncommuting quantities \cite{H_2001,M_1998,S_1994}. The incompatibility of quantum variables and the nature of measurement (as an interaction with a macroscopic apparatus which affects the quantum  system) impose restrictions on information which can be retrieved without disturbing the system.

Nevertheless, for a class of open quantum systems which are weakly coupled to external electromagnetic fields, the  measurement of the output fields can be arranged in a nondemolition manner \cite{B_1989} so that, at any moment of time, the past observations commute between themselves and with any future system operator. In this case, the quantum measurements are, in many respects, similar to classical observations and, in fact, can be regarded as classical random processes \cite{KS_1991} on a common probability space. The statistical dependence on the system variables,  which results from the system-field interaction, allows such observations to be used for continuously updating the conditional density operator of the quantum system according to the stochastic master equation (SME) \cite{GZ_2004,WM_2010}. This posterior density operator and its modifications plays the role of an information state in measurement-based quantum control  problems \cite{B_1983,DDJW_2006,EB_2005,GBS_2005}. %,J_2005,YB_2009}.
The SME is a quantum analogue of the Kushner-Stratonovich equation \cite{LS_2001} for the evolution of the posterior probability density function (PDF) of the system variables in the case of classical system-observation dynamics described by stochastic differential equations (SDEs). Similarly to its classical counterpart, the SME is a recursive implementation of  the Bayesian inference. % \cite{LS_2001}.
Accordingly, the quantum Belavkin-Kushner-Stratonovich equation (BKSE) \cite{BVJ_2007,EWP_2015,GK_2010,N_2014},  which governs the dynamics of the conditional expectations of system operators, %(and is, therefore, dual to the SME),
is developed in the framework of the Hudson-Parthasarathy calculus of quantum stochastic differential equations (QSDEs) \cite{HP_1984,P_1992}.

In the QSDE model of open quantum systems, the external input bosonic fields are accommodated by (and have their own quantum state on) a symmetric Fock space \cite{PS_1972}. These fields   are represented by quantum Wiener processes on the Fock space which drive the system variables according to the energetics of the system and its interaction with the fields. The latter is specified by the system Hamiltonian and the system-field coupling operators which are functions of the system variables.  However, classical functions of several real or complex variables can be extended to the noncommutative quantum variables in different ways. One of such extensions is provided by the Weyl functional calculus \cite{F_1989} which employs unitary Weyl operators whose role in this context is similar to that of the spatial harmonics in the Fourier transform.

The Weyl quantization is used in the Wigner-Moyal phase-space method \cite{Hi_2010,M_1949} of quasi-probability density functions  (QPDFs) which are the Fourier transforms of the quasi-characteristic functions (QCFs) \cite{CH_1971}, with the latter being the quantum expectations of the Weyl operators. The phase-space approach allows the quantum dynamics to be represented without the ``burden of the Hilbert space'' and leads to partial differential and integro-differential equations for the QPDFs and QCFs, which involve only real or complex variables and encode the moments of the system operators. Although the Moyal equations  \cite{M_1949} for the QPDF dynamics were originally obtained for closed systems, the phase-space approach  has also extensions to different classes of open quantum systems; see, for example, \cite{GRS_2014,GRS_2015,KS_2008,MD_2015,V_2015c}.

In the present paper, the phase-space approach is applied to the filtering problem for a class of nonlinear quantum stochastic systems with multichannel field measurements satisfying the nondemolition conditions. The system variables satisfy the Weyl canonical commutation relations (CCRs) and are governed by a Markovian  QSDE driven by the quantum Wiener processes of bosonic fields in vacuum state. Using the Weyl quantization of the Hamiltonian and system-field coupling operators in combination with the results of \cite{N_2014} and \cite{V_2015c},  we obtain a stochastic integro-differential equation (SIDE) for the evolution of the  posterior QCF  of the system conditioned on the measurements. This equation is a spatial frequency domain representation of the BKSE driven by the innovation process associated with the measurements. We also discuss a more specific form of the posterior QCF and QPDF dynamics for a class \cite{SVP_2014,V_2015c}  of open quantum systems whose coupling operators are linear functions of the system variables while the Hamiltonian is split into a quadratic part and a nonquadratic part represented in the Weyl quantization form. For this linear system-field coupling case, we outline  modified quantum Kalman filter equations  for a Gaussian approximation of the posterior system state. We also mention that the Weyl quantization of the Hamiltonian and coupling operators has recently been used in \cite{V_2015a,V_2015b} in a different context of optimality conditions for coherent (measurement-free) quantum control and filtering problems.

The paper is organized as follows. Section~\ref{sec:sys} describes the class of quantum stochastic systems under consideration. Section~\ref{sec:meas} specifies the model of nondemolition measurements and describes the BKSE for conditional expectations. Section~\ref{sec:SIDE} applies this equation to the Weyl operators and obtains the posterior QCF dynamics in the Weyl quantization framework. Section~\ref{sec:hlin} specifies these results, together with a related equation for the posterior QPDF, for the case of linear system-field coupling. Section~\ref{sec:gauss} develops modified quantum Kalman filter equations for a Gaussian approximation of the posterior quantum state.
Section~\ref{sec:conc} provides concluding remarks.

%%%%%%%%%%%%%%%%%%%%%%%%%%%%%%%%%%%%%%%%%%%%%%%%%%%%%%%%%%%%%%%%%%%%%%%%%%%%%%%%%%%%%%%%%%%%%%%%%%%
\section{QUANTUM STOCHASTIC SYSTEMS BEING CONSIDERED}
\label{sec:sys}
%%%%%%%%%%%%%%%%%%%%%%%%%%%%%%%%%%%%%%%%%%%%%%%%%%%%%%%%%%%%%%%%%%%%%%%%%%%%%%%%%%%%%%%%%%%%%%%%%%%

We consider an open quantum system, whose internal dynamics are affected by interaction with external fields and are described in terms of an even number $n$ of dynamic variables $X_1,\ldots, X_n$ assembled into a vector $X:= (X_k)_{1\< k\< n}$ (vectors are assumed to be organized as columns). These   system variables are time-varying self-adjoint operators on a complex separable Hilbert space $\cH$ satisfying the CCRs
\begin{equation}
\label{CCR}
    \cW_{u+v} = \re^{i u^{\rT}\Theta v} \cW_u \cW_v
\end{equation}
for all $ u,v\in \mR^n$ in terms of the unitary Weyl operators \cite{F_1989}
\begin{equation}
\label{cW}
  \cW_u := \re^{iu^{\rT} X} = \cW_{-u}^{\dagger},
\end{equation}
where $(\cdot)^{\dagger}$ is the operator adjoint. Here, $\Theta$ is a nonsingular real antisymmetric matrix which specifies the matrix of commutators
$
    [X, X^{\rT}]
    :=    ([X_j,X_k])_{1\< j,k\< n}
     =
     2i \Theta
$
as an infinitesimal form of the Weyl CCRs (\ref{CCR}) (the transpose $(\cdot)^{\rT}$ acts on matrices of operators as if their entries were scalars).

For example, in the case when the system variables consist of conjugate pairs of the quantum mechanical position and momentum operators \cite{M_1998} comprising the $\frac{n}{2}$-dimensional vectors $q$ and $p:= -i\d_q$  (with the reduced Planck constant set to $\hslash = 1$), so that
\begin{equation}
\label{Xqp}
    X:= {\begin{bmatrix}q\\ p\end{bmatrix}},
\end{equation}
the CCR matrix takes the form
$    \Theta = \frac{1}{2}    {\small\begin{bmatrix}
        0 & 1\\
        -1 & 0
    \end{bmatrix}}\ox I_{n/2}
$
and corresponds to the symplectic structure matrix in classical Hamiltonian systems (here, $\ox$ denotes the Kronecker product of matrices, and $I_r$ is the identity matrix of order $r$).

The evolution of the vector $X$ of system variables is governed by a Markovian Hudson-Parthasarathy QSDE with the identity scattering matrix \cite{HP_1984,P_1992}
\begin{equation}
\label{dX}
    \rd X =
    \cL(X)\rd t  - i[X,h^{\rT}]\rd W
\end{equation}
whose structure is described below (the time arguments are often omitted for brevity). Although it resembles classical SDEs \cite{KS_1991}, the QSDE (\ref{dX}) is driven by a vector $W:= (W_k)_{1\< k \< m}$ of an even number $m$ of self-adjoint quantum Wiener processes $W_1, \ldots, W_m$ acting on a symmetric Fock space $\cF$. These  represent the external bosonic fields \cite{H_1991,P_1992} and have the quantum Ito table
\begin{equation}
\label{Omega}
    \rd W\rd W^{\rT} = \Omega \rd t,
    \quad
        \Omega
    :=
    I_m + iJ,
    \quad
        J :=
    {\begin{bmatrix}0 & 1\\ -1 & 0\end{bmatrix}} \ox I_{m/2}.
\end{equation}
In contrast to the diffusion matrix of the standard Wiener process, $\Omega := (\omega_{jk})_{1\< j,k\< m}$ is  a complex positive semi-definite Hermitian matrix with a nonzero imaginary part   $\Im \Omega=J$, whereby the quantum Wiener processes  $W_1, \ldots, W_m$ do not commute with each other:
\begin{equation}
\label{WWst}
    [W(s), W(t)^{\rT}] = 2i\min(s,t)J ,
    \qquad s,t\>0.
\end{equation}
Furthermore, the $n$-dimensional
drift vector
$    \cL(X)
$
and the dispersion $(n\x m)$-matrix $- i[X,h^{\rT}]$ of the QSDE (\ref{dX}) are specified by the system Hamiltonian $h_0$ and the vector   $
    h:=
    (h_k)_{1\<k\<m}$ of
system-field coupling operators $h_1, \ldots, h_m$, which are self-adjoint operators on $\cH$ representable as functions of the system variables.
The superoperator $\cL$ is the Gorini-Kossakowski-Sudar\-shan-Lindblad (GKSL) generator \cite{GKS_1976,L_1976} which acts on a system operator $\xi$ (and applies entrywise to vectors) as
\begin{equation}
\label{cL}
   \cL(\xi):= i[h_0,\xi] +    \frac{1}{2}
    \sum_{j,k=1}^m
    \omega_{jk}
    \big( [h_j,\xi]h_k + h_j[\xi,h_k]\big).
\end{equation}
The specific structure of the QSDE (\ref{dX}) comes from  the system-field interaction which drives a unitary operator $U(t)$ acting on the system-field tensor-product space $\cH:= \cH_0\ox \cF$ (with $\cH_0$ the initial space for the action of the system variables at time $t=0$):
\begin{equation}\label{dU}
    \rd U = -U \Big(i(h_0\rd t + h^{\rT} \rd W) + \frac{1}{2}h^{\rT}\Omega h\rd t\Big),
\end{equation}
where $U(0)=\cI_{\cH}$  is the identity operator on $\cH$. The unitary operator $U(t)$, which  depends on the system-field interaction over the time interval from $0$ to $t$, is adapted in the sense that  it acts effectively on the subspace $\cH_0\ox \cF_t$, where $\{\cF_t:\, t\>0\}$ is the Fock space filtration.   The corresponding quantum stochastic flow evolves the system variables as
\begin{equation}
\label{uni}
    X(t)
    =
    U(t)^{\dagger} (X(0)\ox \cI_{\cF}) U(t),
\end{equation}
whence the QSDE (\ref{dX}) is obtained by using (\ref{dU}) and the quantum Ito formula \cite{HP_1984,P_1992} in combination with (\ref{Omega}) and commutativity between the forward Ito increments $\rd W(t)$ and adapted processes taken at time $s\< t$. Adapted processes $\xi$, which are functions  of the system variables, satisfy QSDEs of the same form
\begin{equation}
\label{dxi}
    \rd \xi
    =
    \cL(\xi)\rd t - i[\xi,h^{\rT} ]\rd W.
\end{equation}

%%%%%%%%%%%%%%%%%%%%%%%%%%%%%%%%%%%%%%%%%%%%%%%%%%%%%%%%%%%%%%%%%%%%%%%%%%%%%%%%%%%%%%%%%%%%%%%%%%%
\section{OUTPUT FIELDS, NONDEMOLITION MEASUREMENTS AND CONDITIONING}\label{sec:meas}
%%%%%%%%%%%%%%%%%%%%%%%%%%%%%%%%%%%%%%%%%%%%%%%%%%%%%%%%%%%%%%%%%%%%%%%%%%%%%%%%%%%%%%%%%%%%%%%%%%%

A more conventional representation of the quantum Wiener process $W$ is in terms of
the field annihilation  $\fa_1, \ldots, \fa_{m/2}$ and creation $\fa_1^{\dagger}, \ldots, \fa_{m/2}^{\dagger}$  processes \cite{HP_1984,P_1992}, assembled into vectors $\fa:= (\fa_k)_{1\<k\<m/2}$ and $\fa^{\#}:= (\fa_k^{\dagger})_{1\<k\<m/2}$:
\begin{equation}
\label{Wfa}
    W
    :=
    2
    {\begin{bmatrix}
        \Re \fa\\
        \Im \fa
    \end{bmatrix}}
    =
    \left(
    {\begin{bmatrix}
        1 & 1\\
        -i & i
    \end{bmatrix}}
    \ox I_{m/2}
    \right)
    {\begin{bmatrix}
        \fa \\
        \fa^{\#}
    \end{bmatrix}}.
\end{equation}
Here, $(\cdot)^{\#}$ denotes the entrywise operator adjoint, and the real and imaginary parts %of a complex matrix
are extended to matrices $M$ with operator-valued entries as $\Re M = \frac{1}{2}(M+M^{\#})$ and $\Im M = \frac{1}{2i}(M-M^{\#})$ which consist of self-adjoint operators. In accordance with (\ref{Omega}),
the quantum Ito table of the annihilation and creation processes is given by
$
    \rd {\small\begin{bmatrix}\fa \\ \fa^{\#}\end{bmatrix}}
    \rd {\small\begin{bmatrix}\fa^{\dagger} & \fa^{\rT}\end{bmatrix}}
%    :=
%    \begin{bmatrix}
%        \rd\fa \rd \fa^{\dagger} & \rd \fa \rd \fa^{\rT}\\
%        \rd\fa^{\#} \rd \fa^{\dagger} & \rd \fa^{\#} \rd \fa^{\rT}
%    \end{bmatrix}
    =
    {\small\begin{bmatrix}
        1 & 0\\
        0 & 0
    \end{bmatrix}}
    \ox I_{m/2}
    \rd t
$,
where $(\cdot)^{\dagger}:= ((\cdot)^{\#})^{\rT}$ denotes the transpose of the entrywise adjoint and reduces to the complex conjugate transpose $(\cdot)^*:= (\overline{(\cdot)})^{\rT}$ for complex matrices. Also, the term
$
    ih^{\rT}\rd W = L^{\dagger}\rd \fa - L^{\rT}\fa^{\#}
$,
which is part of the diffusion term in (\ref{dU}),
is related to a different vector $L:=(L_k)_{1\<k\<m/2}$ of (not necessarily self-adjoint) coupling operators $L_1, \ldots, L_{m/2}$, so that
\begin{equation}
\label{hL}
    h =
        -J
        {\begin{bmatrix}
          \Re L \\
          \Im L
        \end{bmatrix}},
        \qquad
        {\begin{bmatrix}
          \Re L \\
          \Im L
        \end{bmatrix}}
        =
        Jh.
\end{equation}
Here, use is also made of the identity $J^2 = -I_m$ following from (\ref{Omega}).
The relations (\ref{Wfa}) and (\ref{hL}) make it possible to move between two alternative representations of the external fields and the system-field coupling operators.
As a result of the joint system-field evolution described by the unitary operator $U(t)$ from (\ref{dU}), the output field $Y:= (Y_k)_{1\<k\<m}$ is given by
\begin{equation}
\label{Y}
    Y(t) =
    2
    {\begin{bmatrix}
        \Re \fb(t)\\
        \Im \fb(t)
    \end{bmatrix}}
    =
 U(t)^{\dagger}(\cI_{\cH_0}\ox W(t))U(t)
\end{equation}
and satisfies the QSDE
\begin{equation}
\label{dY}
  \rd Y = 2Jh\rd t + \rd W.
\end{equation}
Here, $\fb:= (\fb_k)_{1\<k\<m/2}$ and $\fb^{\#}:= (\fb_k^{\dagger})_{1\<k\<m/2}$ are vectors of the corresponding output annihilation and creation operators:
\begin{equation}
\label{fb}
    \fb   = U^{\dagger}(\cI_{\cH_0}\ox \fa)U,
    \qquad
    \fb^{\#}   = U^{\dagger}(\cI_{\cH_0}\ox \fa^{\#})U.
\end{equation}
In view of (\ref{Wfa}), (\ref{hL}) and (\ref{dY}), the processes $\fb$ and $\fb^{\#}$ satisfy
the QSDEs
$  \rd \fb = L\rd t + \rd \fa$ and
$
    \rd \fb^{\#} = L^{\#}\rd t + \rd \fa^{\#}
$,
the second of which is obtained by the conjugation of the first one.
The unitary evolution in (\ref{uni}) and (\ref{Y}) preserves the commutativity between the system and output field variables in the sense that
\begin{equation}
\label{XY}
        [X(t),Y(s)^{\rT}]
     =
    0,
    \qquad
    t\> s\> 0.
\end{equation}
However, since
\begin{equation}
\label{YY}
    [\rd Y, \rd Y^{\rT}]
    =
    [\rd W, \rd W^{\rT}] = 2iJ\rd t
\end{equation}
and $[Y(s), Y(t)^{\rT}] = 2i\min(s,t)J$ for all $s,t\>0$ in view of (\ref{dY}) and (\ref{WWst}), the output fields $Y_1, \ldots, Y_m$ do not commute with each other and are not accessible to simultaneous measurement.

Therefore, following \cite{N_2014},  we will consider an  $r$-channel field $Z$ which is related to $\fb$ and $\fb^{\#}$ from (\ref{fb}) and $Y$ from (\ref{Y}) by
\begin{equation}
\label{Z}
    Z = \overline{G}\fb + G\fb^{\#} = 2\Re(\overline{G}\fb) = F Y,
    \quad\
    F :=
    \begin{bmatrix}
      \Re G  & \Im G
    \end{bmatrix}.\!\!\!
\end{equation}
Here, it is assumed that $r\< \frac{m}{2}$, and the matrix $F\in \mR^{r\x m}$, which is specified by $G\in \mC^{r\x m/2}$,  satisfies the conditions
\begin{equation}
\label{FF}
  FF^{\rT} \succ 0,
  \qquad
  FJF^{\rT} = 0,
\end{equation}
the first of which is equivalent to $F$ being of full row rank.
  In view of (\ref{YY}),  the second condition in (\ref{FF}) implies that $[\rd Z, \rd Z^{\rT}] = 2iF JF^{\rT}\rd t=0$,  which makes the quantum process $Z$ in (\ref{Z})  self-commuting and allows for simultaneous continuous measurements of its entries $Z_1, \ldots, Z_r$. Furthermore, $Z$ can be regarded (up to an isomorphism) as a classical  diffusion process \cite{KS_1991} with values in  $\mR^r$ and a positive definite diffusion matrix $FF^{\rT}$. Also, $Z$ inherits from $Y$ the property (\ref{XY}) since
$
    [X(t),Z(s)^{\rT}] = [X(t),Y(s)^{\rT}] F^{\rT} = 0
$ for all
$
    t\> s\>0
$.
Hence, for any time $t\>0$, the past measurement history
\begin{equation}
\label{Zt}
    \fZ_t:= \left\{Z_1(s), \ldots, Z_r(s):\ 0\<s\< t\right\}
\end{equation}
and any given system operator $\xi(t):=f(X(t))$ (an operator-valued extension of a  complex-valued function to the system variables) form a set of pairwise commuting (and hence, compatible) quantum variables. This makes the process $Z$ in (\ref{Z}) and (\ref{FF}) a legitimate model of nondemolition measurements.
In what follows, we will use the conditional quantum expectation
\begin{equation}
\label{pi}
    \pi_t(\xi)
    :=
    \bE(\xi(t)\mid \cZ_t)
\end{equation}
of a system operator $\xi$ at time $t\>0$ with respect to the commutative von Neumann algebra $\cZ_t$ generated by the past measurement history $\fZ_t$  from (\ref{Zt}). This is a mean square optimal estimator of $\xi(t)$   in the sense that $\pi_t(\xi)$ is an element of the measurement algebra $\cZ_t$ which delivers the minimum
\begin{equation}
\label{Emin}
    \min_{\eta \in \cZ_t}
    \bE((\xi(t)-\eta)^{\dagger}(\xi(t)-\eta)).
\end{equation}
This characterization is
similar  to the variational property of classical conditional expectations (of square integrable random variables)  with respect to $\sigma$-subalgebras \cite{LS_2001}.
The quantum expectation $\bE \zeta = \Tr(\rho \zeta)$ in (\ref{Emin}) is over the system-field density operator $\rho := \varpi \ox \ups$, where $\varpi$ is the initial quantum state of the system, and $\ups$ is the vacuum state \cite{P_1992} of the input fields.

According to \cite[Theorem 9]{N_2014}, the conditional expectation of a given system operator $\xi$ in (\ref{pi}) with respect to the nondemolition measurements (\ref{Z}) satisfies the BKSE
\begin{equation}
\label{dpi}
    \rd \pi_t(\xi)
    =
    \pi_t(\cL(\xi))
    \rd t
    +
    \beta^{\rT}
    K
    \rd \chi,
\end{equation}
which is driven by an innovation process $\chi$ (a martingale with respect to the measurement filtration) with the Ito differential
\begin{equation}
\label{dchi}
    \rd \chi
    =
    \rd Z - 2FF^{\rT}K^{\rT}\pi_t(\Re M)\rd t
\end{equation}
and diffusion matrix $FF^{\rT}$.
Here, the conditional expectation $\pi_t$ is evaluated at vectors of system operators entrywise, and
\begin{equation}
\label{beta}
    \beta
    :=
        \pi_t(M^{\#} \xi + \xi M)
        -2
        \pi_t(\xi)
        \pi_t(\Re M).
\end{equation}
The drift term of the SDE (\ref{dpi}) comes from $\xi$ having dynamics of its own in (\ref{dxi}).  The diffusion term $    \beta^{\rT}
    K
    \rd \chi$ represents the measurement-driven corrections of the prior estimate and, together with (\ref{dchi}) and (\ref{beta}), involves additional quantities \cite{N_2014}. More precisely,
\begin{equation}
\label{ML}
    M:= E^{-\rT}L
\end{equation}
(with $E^{-\rT}:= (E^{-1})^{\rT}$)
is a modified vector of coupling operators, where $E\in \mC^{\frac{m}{2}\x\frac{m}{2}} $ is a nonsingular matrix  obtained by augmenting the matrix $G$ from (\ref{Z}) as
\begin{equation}
\label{E}
  E:= {\begin{bmatrix}
        G \\
        D
      \end{bmatrix}}.
\end{equation}
The matrix $D \in \mC^{(\frac{m}{2}-r)\x\frac{m}{2}} $ is found so as to satisfy the condition
\begin{equation}
\label{EE}
        \begin{bmatrix}
      \Re E  & \Im E
    \end{bmatrix}J
    {\begin{bmatrix}
      \Re E^{\rT}  \\ \Im E^{\rT}
    \end{bmatrix}} = 0.
\end{equation}
In view of (\ref{E}), the leading diagonal block of order $\frac{m}{2}$ in the matrix on the left-hand side of (\ref{EE}) is $FJF^{\rT}$ which vanishes due to the second condition in (\ref{FF}). The matrix $K \in \mR^{\frac{m}{2}\x r}$ in (\ref{dpi}) and (\ref{dchi}) is computed according to \cite[Lemma 6, Proposition 8, Remark 11]{N_2014} as
\begin{equation}
\label{K}
    K:=
    \begin{bmatrix}
      I_r \\
      (\Re D \Re G^{\rT} + \Im D\Im G^{\rT})(FF^{\rT})^{-1}
    \end{bmatrix}.
\end{equation}
Whereas the SDE (\ref{dpi}) follows the Heisenberg picture  of quantum dynamics, its dual Schr\"{o}dinger picture version,  known as the SME \cite{WM_2010},  describes the evolution of the posterior density operator. The latter is a quantum counterpart of the classical conditional probability distribution which is continuously updated over the course of measurements according to the Bayes rule. However, in contrast to the classical case, the SME approach carries the ``burden of the Hilbert space''.

At the same time, the SDE (\ref{dpi}) is not algebraically closed, in general,  since its right-hand side involves other conditional moments which are not necessarily reducible to $\pi_t(\xi)$. The desired closure can be achieved within an appropriate parametric family of system operators $\xi$.
Such family is provided, for example, by the Weyl operators (\ref{cW}) in the Weyl quantization framework, which is considered in the next section.

%%%%%%%%%%%%%%%%%%%%%%%%%%%%%%%%%%%%%%%%%%%%%%%%%%%%%%%%%%%%%%%%%%%%%%%%%%%%%%%%%%%%%%%%%%%%%%%%%%%
\section{EVOLUTION OF THE POSTERIOR QUASI-CHARACTERISTIC FUNCTION}\label{sec:SIDE}
%%%%%%%%%%%%%%%%%%%%%%%%%%%%%%%%%%%%%%%%%%%%%%%%%%%%%%%%%%%%%%%%%%%%%%%%%%%%%%%%%%%%%%%%%%%%%%%%%%%

Application of the conditional expectation (\ref{pi}) to the Weyl operator $\cW_u$, associated with the system variables by (\ref{cW}), leads to the posterior QCF
\begin{equation}
\label{Phi}
    \Phi(t,u):= \pi_t(\cW_u) = \overline{\Phi(t,-u)},
    \qquad
    t\>0,\
    u \in \mR^n,
\end{equation}
which is a complex-valued function.
Here, the second equality  describes the Hermitian property of $\Phi(t,u)$ with respect to its spatial argument $u$ and follows from the second equality in (\ref{cW}), while $\cW_0 = \cI_{\cH}$ implies that $\Phi(t,0)=1$. However, in view of the Weyl CCRs (\ref{CCR}), the Bochner-Khinchin positiveness criterion \cite{GS_2004}  for the characteristic functions of classical probability distributions is replaced with its quantum mechanical weighted version \cite{CH_1971,H_2010}: the complex Hermitian matrix $\big(\re^{iu_j^{\rT}\Theta u_k} \Phi(t,u_j-u_k)\big)_{1\<j,k\<\ell}$ is positive semi-definite for any number of arbitrary points $u_1, \ldots, u_{\ell} \in \mR^n$.

The spatial Fourier transform of (\ref{Phi}) yields a real-valued posterior QPDF
\begin{equation}
\label{mho}
    \mho(t,x):= (2\pi)^{-n}\int_{\mR^n} \Phi(t,u)\re^{-iu^{\rT}x}\rd u,
    \qquad
    x \in \mR^n.
\end{equation}
Although the function $\mho$ is not necessarily nonnegative everywhere \cite{H_1974} (since, as mentioned above, the QCF $\Phi$ does not have to be positive semi-definite), it satisfies the normalization condition $    \int_{\mR^n}\mho(t,x)\rd x = \Phi(t,0)=1$ and is a quantum analogue of the classical posterior PDF. In particular, $\Phi$ and $\mho$  encode  information on conditional moments of the system variables $\pi_t(X_{k_1}\x \ldots \x X_{k_{\ell}})$ for any $1\< k_1,\ldots, k_{\ell}\< n$, provided $\Phi(t,u)$ is $\ell$ times continuously differentiable  with respect to $u\in \mR^n$.  Up to an isomorphism,  both functions $\Phi(t,\cdot)$ and $\mho(t,\cdot)$ are $\cZ_t$-adapted classical random fields \cite{GS_2004} on $\mR^n$ on a common probability space.

Since the posterior QCF (\ref{Phi}) is the conditional expectation of the Weyl operators, the BKSE (\ref{dpi}) applies to this case too. Moreover, this leads to an algebraically  closed equation for the time evolution of the posterior QCF in the framework of the Weyl quantization model for the energy operators of the system.
To this end, following \cite{V_2015c}, we assume that the system Hamiltonian $h_0$ and the system-field coupling operators $h_1, \ldots, h_m$ in (\ref{cL}) are obtained by the Weyl quantization  \cite{F_1989} of real-valued functions on $\mR^n$ with the Fourier transforms $H_k: \mR^n\to \mC$ as
\begin{equation}
\label{hk}
    h_k := \int_{\mR^n} H_k(u)\cW_u \rd u,
    \qquad
    k = 0,1,\ldots, m,
\end{equation}
where $\cW_u$ is the Weyl operator (\ref{cW}).
The vector $h$ of the system-field coupling operators is related to the vector-valued map $H:= (H_k)_{1\< k\<m}: \mR^n \to \mC^m$ by
\begin{equation}\label{h}
    h = \int_{\mR^n} H(u)\cW_u\rd u.
\end{equation}
The operators $h_k$ in (\ref{hk})   are self-adjoint due to the second equality in (\ref{cW}) and the Hermitian property of the functions $H_k$ as the Fourier transforms of real-valued functions
 (that is, $H_k(-u) = \overline{H_k(u)}$ for all $u\in \mR^n$).

If the function $H_k$ is absolutely integrable (that is, $\int_{\mR^n} |H_k(u)|\rd u<+\infty$), then (\ref{hk}) can be understood as a Bochner integral  \cite{Y_1980} which yields a bounded operator $h_k$ due to the unitarity of the Weyl operator $\cW_u$ for any $u \in \mR^n$. In order to obtain, for example, polynomial functions of the system variables,  the Fourier transforms $H_k$ in (\ref{hk}) have to be  generalized functions \cite{V_2002} (a particular class of such systems will be considered in Section~\ref{sec:hlin}).

The following theorem uses integral operators $\cA$, $\cB$, $\cC$ (of which $\cA$ and $\cC$ are linear) which map a function $\varphi: \mR^n \to \mC$ to the functions $\cA(\varphi): \mR^n\to \mC$ and $\cB(\varphi), \cC(\varphi): \mR^n\to \mC^{\frac{m}{2}}$ as
\begin{align}
\label{cA}
    \cA(\varphi)(u)
    & :=
    \int_{\mR^n}
    V(u,v)\varphi(u+v)
    \rd v,\\
\label{cB}
  \cB(\varphi)(u)
  & :=
  \cC(\varphi)(u)
  -
    \varphi(u) \cC(\varphi)(0),\\
\label{cC}
      \cC(\varphi)(u)
       & :=
        \int_{\mR^n}
            \Gamma(u,v)
            \varphi(u+v)
            \rd v.
\end{align}
The kernel function $V: \mR^n \x \mR^n\to \mC$  in (\ref{cA}) is computed as
\begin{equation}
\label{V}
  V(u,v)
  :=
    -2        \sin(u^{\rT}\Theta v)H_0(v)
     -2 \int_{\mR^n}
            \sin(u^{\rT}\Theta w)
            H(w)^{\rT}
            \Ups(u+w,w-v)
        H(v-w) \rd w,
\end{equation}
where $H_0$ and $H$  are the Fourier transforms from (\ref{hk}) and (\ref{h}), and
$\Ups: \mR^n\x \mR^n\to \mR^{m\x m}$ is an auxiliary function which is expressed as
\begin{equation}
\label{Ups}
    \Ups(u,v)
     :=
    \sin(u^{\rT}\Theta v) I_m + \cos(u^{\rT}\Theta v) J
\end{equation}
in terms of the CCR matrix $\Theta$ in (\ref{CCR}) and the matrix $J$ in (\ref{Omega}). Also, the function $\Gamma: \mR^n\x \mR^n\to \mR^{\frac{m}{2}}$ in (\ref{cC}) is related by
\begin{align}
\label{Gam}
    \Gamma(u,v)
        :=
        2
        \Big(\cos(u^{\rT}\Theta v)&
        {\begin{bmatrix}
            E_2 & E_1
        \end{bmatrix}}
         +
        \sin(u^{\rT}\Theta v)
        {\begin{bmatrix}
            -E_1 & E_2
        \end{bmatrix}}
        \Big)
        H(v)
\end{align}
to the matrix $E$ from (\ref{E}) and (\ref{EE}) through the matrices $E_1, E_2 \in \mR^{\frac{m}{2}\x \frac{m}{2}}$ given by
\begin{equation}
\label{E12}
                           E_1:= \Re(E^{-\rT}),
                           \qquad
                           E_2:= \Im(E^{-\rT}).
\end{equation}

%%%%%%%%%%%%%%%%%%%%%%%%%%%%%%%%%%%%%%%%%%%%%%%%%%%%%%%%%%%%%%%%%%%%%%%%%%%%%%%%%%%%%%%%%%%%%%%%%%%%%%%%%%%%%%%
\begin{thm}
\label{th:Phidot}
Suppose the Hamiltonian $h_0$ and the coupling operators $h_1, \ldots, h_m$ of the quantum stochastic system (\ref{dX}) have the Weyl quantization form (\ref{hk}). Then the posterior QCF (\ref{Phi}) with respect to the nondemolition measurements in (\ref{Z}), (\ref{FF}) satisfies the SIDE
\begin{equation}
\label{dPhi}
    \rd\Phi(t,u)
    =
    \cA(\Phi(t,\cdot))(u)
    \rd t
    +
    \cB(\Phi(t,\cdot))(u)^{\rT}K\rd \chi.
\end{equation}
Here, the innovation process $\chi$ does not depend on $u\in \mR^n$ and its Ito differential is given by
\begin{equation}
\label{dchiPhi}
    \rd \chi
    =
        \rd Z
        -
        FF^{\rT}K^{\rT}
        \cC(\Phi(t,\cdot))(0)\rd t,
\end{equation}
where the integral operators  $\cA$, $\cB$, $\cC$  from (\ref{cA})--(\ref{cC}) act over the spatial argument of $\Phi$, and the matrix $K\in \mR^{\frac{m}{2}\x r}$ is given by (\ref{K}).
\hfill$\square$
\end{thm}
%%%%%%%%%%%%%%%%%%%%%%%%%%%%%%%%%%%%%%%%%%%%%%%%%%%%%%%%%%%%%%%%%%%%%%%%%%%%%%%%%%%%%%%%%%%%%%%%%%%%%%%%%%%%%%%
\begin{proof}
We will evaluate the terms of the BKSE (\ref{dpi}) at the Weyl operator $\xi:= \cW_u$ in (\ref{cW}) using the Weyl quantization (\ref{hk}) and (\ref{h}).
From the proof of \cite[Theorem 1]{V_2015c}, it follows that the GKSL generator (\ref{cL}) acts on $\cW_u$ as
\begin{equation}
\label{cLcW}
    \cL(\cW_u) =      \int_{\mR^n}
    V(u,v)\cW_{u+v}
    \rd v,
\end{equation}
where the function $V$ is computed according to (\ref{V}) and (\ref{Ups}). In view of (\ref{Phi}), the conditional expectation of (\ref{cLcW}) takes the form
\begin{equation}
\label{picLW}
    \pi_t(\cL(\cW_u))
     =
%    \int_{\mR^n}
%    V(u,v)\pi_t(\cW_{u+v})
%    \rd v\\
    \int_{\mR^n}
    V(u,v)\Phi(t,u+v)
    \rd v
     =
    \cA(\Phi(t,\cdot))(u),
\end{equation}
with $\cA$ given by (\ref{cA}).
The modified vector of coupling operators in (\ref{ML}) can be represented as
\begin{align}
\nonumber
    M
    & =
    E^{-\rT}(\Re L + i\Im L)
     =
    \left({\begin{bmatrix}
            1 & i
          \end{bmatrix}}
          \ox E^{-\rT}
           \right)
          {\begin{bmatrix}
           \Re L\\
           \Im L
         \end{bmatrix}}\\
\label{M}
    & =
    \left({\begin{bmatrix}
            1 & i
          \end{bmatrix}}
          \ox E^{-\rT}
           \right)
           Jh
     =
    \left({\begin{bmatrix}
            1 & i \\
          \end{bmatrix}}
          \ox E^{-\rT}
           \right)
           J
        \int_{\mR^n} H(v)\cW_v\rd v,
\end{align}
where use is made of (\ref{hL}) and the Weyl quantization (\ref{h}) of the coupling operators. In view of the Weyl CCRs (\ref{CCR}), it follows from (\ref{M}) that
\begin{align}
\nonumber
    \cW_u M
    & =
    \left({\begin{bmatrix}
            1 & i
          \end{bmatrix}}
          \ox E^{-\rT}
           \right)
           J
          \int_{\mR^n} H(v)\cW_u\cW_v\rd v\\
\label{WM}
    & =
    \left({\begin{bmatrix}
            1 & i \\
          \end{bmatrix}}
          \ox E^{-\rT}
           \right)
           J
        \int_{\mR^n} \re^{-iu^{\rT}\Theta v}H(v)\cW_{u+v}\rd v.
\end{align}
A similar reasoning leads to
\begin{equation}
\label{MW}
        M^{\#} \cW_u
     =
     \overline{
    \left({\begin{bmatrix}
            1 & i
          \end{bmatrix}}
          \ox E^{-\rT}
           \right)}
           J
        \int_{\mR^n} \re^{iu^{\rT}\Theta v}H(v)\cW_{u+v}\rd v.
\end{equation}
The sum of the left-hand sides of (\ref{WM}) and  (\ref{MW}) takes the form
\begin{equation}
\label{WMMW}
    \cW_u M  + M^{\#}\cW_u
     =
    2
    \int_{\mR^n}
        \Re
        \big(
            \re^{-iu^{\rT}\Theta v}
            \begin{bmatrix}
              1 & i
            \end{bmatrix}
            \ox
            E^{-\rT}
        \big)
        JH(v)
     =
    \int_{\mR^n}
    \Gamma(u,v)
    \cW_{u+v}
    \rd v,
\end{equation}
where use is made of the matrices $J$, $E_1$, $E_2$  from (\ref{Omega}) and (\ref{E12}) leading to the function $\Gamma$ in (\ref{Gam}). The conditional expectation of (\ref{WMMW}) is
\begin{equation}
\label{piWMMW}
    \pi_t(\cW_u M  + M^{\#}\cW_u)
     =
    \int_{\mR^n}
    \Gamma(u,v)
    \Phi(t,u+v)
    \rd v
     =
    \cC(\Phi(t,\cdot))(u),
\end{equation}
with $\cC$ given by (\ref{cC}).
Since the vector $Jh$ in (\ref{M}) consists of self-adjoint operators, then
\begin{equation}
\label{ReM}
    \Re M
     =
    \Re
      \left({\begin{bmatrix}
            1 & i
          \end{bmatrix}}
          \ox E^{-\rT}
           \right)
           J
           h
           =
           {\begin{bmatrix}
            E_1 & -E_2
          \end{bmatrix}}Jh
            =
                      {\begin{bmatrix}
            E_2 & E_1
          \end{bmatrix}}
        \int_{\mR^n} H(v)\cW_v\rd v
\end{equation}
in view of (\ref{E12}).
The conditional expectation of (\ref{ReM}) is related to (\ref{Gam}) and (\ref{cC}) by
\begin{equation}
\label{piReM}
           2\pi_t(\Re M)
           =
        \int_{\mR^n} \Gamma(0,v)\Phi(t,v)\rd v
        =
        \cC(\Phi(t,\cdot))(0).
\end{equation}
Substitution of (\ref{Phi}), (\ref{picLW}), (\ref{piWMMW}) and (\ref{piReM}) into (\ref{dpi})--(\ref{beta}) establishes (\ref{dPhi}) and (\ref{dchiPhi}).
 \end{proof}
%%%%%%%%%%%%%%%%%%%%%%%%%%%%%%%%%%%%%%%%%%%%%%%%%%%%%%%%%%%%%%%%%%%%%%%%%%%%%%%%%%%%%%%

In the absence of measurements, the QCF $\Phi$ is no longer random, the  diffusion term in (\ref{dPhi}) vanishes, and the SIDE reduces to the IDE
$
        \d_t\Phi
    =
    \cA(\Phi(t,\cdot))
$
for the unconditional QCF obtained in \cite[Theorem 1]{V_2015c}. If the system and fields are uncoupled,  this IDE reduces further to the Moyal equation
$\d_t\Phi(t,u) = -2        \int_{\mR^n}\sin(u^{\rT}\Theta v)H_0(v)\Phi(t,u+v)\rd v$, which follows from (\ref{cA}) by letting $H=0$ in (\ref{V}).

%%%%%%%%%%%%%%%%%%%%%%%%%%%%%%%%%%%%%%%%%%%%%%%%%%%%%%%%%%%%%%%%%%%%%%%%%%%%%%%%%%%%%%%%%%%%%%%%%%%
\section{POSTERIOR QCF DYNAMICS IN THE CASE OF LINEAR SYSTEM-FIELD COUPLING}\label{sec:hlin}
%%%%%%%%%%%%%%%%%%%%%%%%%%%%%%%%%%%%%%%%%%%%%%%%%%%%%%%%%%%%%%%%%%%%%%%%%%%%%%%%%%%%%%%%%%%%%%%%%%%

We will now consider a class \cite{SVP_2014,V_2015c} of open quantum systems whose coupling operators are linear functions of the system variables:
\begin{equation}
\label{N}
    h  := NX,
\end{equation}
where $N \in \mR^{m\x n}$ is a coupling matrix, while the Hamiltonian
\begin{align}
\label{h0}
    h_0
    :=
    \frac{1}{2} X^{\rT} R X
    +
    \int_{\mR^d} \Psi(v) \cW_{S^{\rT}v} \rd v
\end{align}
consists of a quadratic part, specified by an energy matrix $R=R^{\rT}\in \mR^{n\x n}$,  and a nonquadratic part represented in the Weyl quantization form. The latter depends on $d\<n $ system variables comprising a vector $S X$ (where $S\in \{0,1\}^{d\x n}$ is a submatrix of a permutation matrix of order $n$) and is determined by the Fourier transform $\Psi: \mR^d \to \mC$ of a real-valued function on $\mR^d$. It is assumed that $\int_{\mR^d}|\Psi(v)||v|\rd v<+\infty$.
For such a system,  the representations  (\ref{hk}) and (\ref{h}) hold with
\begin{align}
\label{HHlin1}
    H_0(u)
     & =
    -\frac{1}{2}\Tr (R\delta''(u)) + \int_{\mR^d}\Psi(v)\delta(u-S^{\rT}v)\rd v,\\
\label{HHlin2}
    H(u)  &= iN\delta'(u),
\end{align}
where $\delta'$ and $\delta''$ are the distributional gradient vector and Hessian matrix of the $n$-dimensional Dirac delta function   $\delta$. Since $SS^{\rT}=I_d$, the matrix $S$ describes an isometry between $\mR^d$ and the subspace $S^{\rT}\mR^d\subset \mR^n$. Therefore,  the integral in (\ref{HHlin1}), as a generalized function \cite{V_2002}, is a complex measure on this subspace with density $\Psi$ (with respect to the $d$-variate Lebesgue measure on $S^{\rT}\mR^d$).
The corresponding QSDE (\ref{dX}) takes the form \cite{V_2015c}
\begin{equation}
\label{dXhlin}
    \rd X
    =
    \Big(
        AX+2i\Theta S^{\rT}\int_{\mR^d} \Psi(v) v\cW_{S^{\rT}v}\rd v
    \Big)\rd t + B\rd W,
\end{equation}
where the matrices $A\in \mR^{n\x n}$ and $B\in \mR^{n\x m}$ are related  to the coupling and energy matrices $N$ and $R$ in (\ref{N}) and (\ref{h0}) by
\begin{equation}
\label{AB}
    A:= 2\Theta (R + N^{\rT}JN),% = 2\Theta R - \frac{1}{2}BJB^{\rT}\Theta^{-1},
    \qquad
    B:= 2\Theta N^{\rT}.
\end{equation}
The nonlinear dependence on the system variables  in the QSDE (\ref{dXhlin}) comes from the nonquadratic part of the Hamiltonian. 
For example, in the case (\ref{Xqp}), when the system has a positive definite mass matrix $\sM$, and  the potential energy is split into a quadratic part $\frac{1}{2}q^{\rT} \sK q$ (with the stiffness matrix $\sK$) and a nonquadratic  part $\int_{\mR^d} \Psi(v) \re^{iv^{\rT}q}\rd v$, with $d=\frac{n}{2}$, the representation (\ref{h0}) holds with $R= \diag(\sK, \sM^{-1})$ and $S= {\small \begin{bmatrix}I_d & 0\end{bmatrix}}$.

%%%%%%%%%%%%%%%%%%%%%%%%%%%%%%%%%%%%%%%%%%%%%%%%%%%%%%%%%%%%%%%%%%%%%%%%%%%%%%%%%%%%%%%%%%%%%%%%%%%
\begin{thm}
\label{th:Hlin}
Suppose the vector $h$ of system-field coupling operators and the system Hamiltonian $h_0$ are given by  (\ref{N}) and (\ref{h0}). Then the SIDE (\ref{dPhi}) for the posterior QCF $\Phi$ in (\ref{Phi}) takes the form
\begin{align}
\nonumber
    \rd  \Phi(t,u)
    = &
    \Big(u^{\rT} A \d_u \Phi(t,u)
        -\frac{1}{2}
        |B^{\rT}u|^2
    \Phi(t,u)
     -
    2
    \int_{\mR^d}
    \sin(u^{\rT}\Theta S^{\rT} v)
    \Psi(v)
    \Phi(t,u+S^{\rT}v)\rd v\Big)\rd t\\
\label{dPhihlin}
        & + 2
        i
        \Big(
                \Phi(t,u)
        {\begin{bmatrix}
            -E_1 & E_2
        \end{bmatrix}}
        N
        \Theta u
         -{\begin{bmatrix}
            E_2 & E_1
        \end{bmatrix}}N (\d_u\Phi(t,u)-\Phi(t,u)\d_u\Phi(t,0))
        \Big)^{\rT}K \rd \chi,
\end{align}
where the matrices $A$ and $B$ are given by (\ref{AB}). The corresponding posterior QPDF $\mho$ in (\ref{mho}) satisfies the SIDE
\begin{align}
\nonumber
    \rd\mho(t,x) =&
    \Big(-\div(\mho(t,x)Ax) + \frac{1}{2}\div^2(\mho(t,x) BB^{\rT})
\\
\nonumber
  &
  -
    2
    \int_{\mR^d}
    \Xi(x,v)\mho(t,x-\Theta S^{\rT}v)
    \rd v\Big)\rd t\\
\label{dmhohlin}
        & +
        2
        \Big(
        {\begin{bmatrix}
            E_1 & -E_2
        \end{bmatrix}}
        N
        \Theta \d_x\mho(t,x)
        +
        \mho(t,x)
        {\begin{bmatrix}
            E_2 & E_1
        \end{bmatrix}}
        N \Big(
            x-\int_{\mR^n}\mho(t,y)y\rd y
            \Big)
        \Big)^{\rT}K \rd \chi,
 \end{align}
 where $\div(\cdot)$ is the divergence operator with respect to $x\in\mR^n$, and the kernel function $\Xi: \mR^n\x \mR^d \to \mR$ is expressed as
 \begin{equation}
 \label{Xi}
    \Xi(x,v)
    :=
    \Re \Psi(v) \sin(v^{\rT} S x)
    +
    \Im \Psi(v) \cos(v^{\rT} S x)
\end{equation}
in terms of  the function $\Psi$ and the matrix $S$ from (\ref{h0}). The Ito differential of the innovation process $\chi$ in (\ref{dchiPhi}) can be represented as
\begin{equation}
\label{dchihlin}
    \rd \chi
    =
        \rd Z
        -
        2
        FF^{\rT}K^{\rT}
        {\begin{bmatrix}
            E_2 & E_1
        \end{bmatrix}}
        N
        \int_{\mR^n}
        \mho(t,y)
        y\rd y\ \rd t.
\end{equation}
\hfill $\square$
\end{thm}
%%%%%%%%%%%%%%%%%%%%%%%%%%%%%%%%%%%%%%%%%%%%%%%%%%%%%%%%%%%%%%%%%%%%%%%%%%%%%%%%%%%%%

\begin{proof}
The drift term in (\ref{dPhihlin}) (and its spatial Fourier transform which is the drift term of (\ref{dmhohlin})) was obtained in \cite[Theorem 2]{V_2015c} and can be established directly  by substituting (\ref{HHlin1}) and (\ref{HHlin2}) into (\ref{V}) and (\ref{cA}). We now turn to the diffusion terms of these SIDEs. By substituting (\ref{HHlin2}) into (\ref{Gam}) and using the relation $f\delta' = f(0)\delta'-f'(0)\delta$ for infinitely differentiable functions $f$ (see, for example, \cite{V_2002}), it follows that (\ref{cC}) takes the form
\begin{align}
\nonumber
    \cC(\varphi)(u)
%    \int_{\mR^n}
%    \Gamma&(u,v)
%    \varphi(u+v)
%    \rd v\\
%\nonumber
     =&
        2
        i
        \int_{\mR^n}
        \Big(\cos(u^{\rT}\Theta v)
        {\begin{bmatrix}
            E_2 & E_1
        \end{bmatrix}}
         +
        \sin(u^{\rT}\Theta v)
        {\begin{bmatrix}
            -E_1 & E_2
        \end{bmatrix}}
        \Big)
        N\delta'(v)
        \varphi(u+v)
        \rd v\\
\label{cCphi}
     =&
        2
        i
        \Big(
                \varphi(u)
        {\begin{bmatrix}
            -E_1 & E_2
        \end{bmatrix}}
        N
        \Theta u
        -{\begin{bmatrix}
            E_2 & E_1
        \end{bmatrix}}N \varphi'(u)
        \Big)
\end{align}
for any bounded smooth function $\varphi: \mR^n \to \mC$.
In particular, at $u=0$,
\begin{align}
\label{cCphi0}
    \cC(\varphi)(0)
     =
        -2
        i
        {\begin{bmatrix}
            E_2 & E_1
        \end{bmatrix}}N \varphi'(0).
\end{align}
Substitution of (\ref{cCphi}) and (\ref{cCphi0}) into (\ref{cB}) leads to
\begin{equation}
\label{cBphi}
    \cB(\varphi)(u)
     =
        2
        i
        \Big(
                \varphi(u)
        {\begin{bmatrix}
            -E_1 & E_2
        \end{bmatrix}}
        N
        \Theta u
         -{\begin{bmatrix}
            E_2 & E_1
        \end{bmatrix}}N (\varphi'(u)-\varphi(u)\varphi'(0))
        \Big).
\end{equation}
In view of (\ref{dPhi}), application of (\ref{cBphi}) to the posterior QCF yields the diffusion term in (\ref{dPhihlin}) whose spatial Fourier transform leads to the diffusion term in (\ref{dmhohlin}). The representation (\ref{dchihlin}) follows from (\ref{dchiPhi}), (\ref{cCphi0}) and the relation $-i\d_u\Phi(t,0) = \int_{\mR^n} \mho(t,y)y\rd y$.
\end{proof}
%%%%%%%%%%%%%%%%%%%%%%%%%%%%%%%%%%%%%%%%%%%%%%%%%%%%%%%%%%%%%%%%%%%%%%%%%%%%%%%%%%%%%

The upper line of (\ref{dmhohlin}) is recognizable as the Fokker-Planck-Kolmogorov equation \cite{KS_1991}
$\d_t\mho =
    -\div(\mho Ax) + \frac{1}{2}\div^2(\mho BB^{\rT})$ for the unconditional PDF of
a classical Markov diffusion process with the linear drift $Ax$ and diffusion matrix $BB^{\rT}$. For the quantum system being considered, the representations
\begin{equation}
\label{piX}
    \pi_t(X)
    =-i\d_u\Phi(t,0) = \int_{\mR^n} \mho(t,x)x\rd x
\end{equation}
for the posterior mean vector of the system variables, similar to the corresponding classical relations,  remain valid even if the QPDF $\mho$ is not nonnegative everywhere. Now, suppose $\Psi=0$, so that the function $\Xi$ in (\ref{Xi}) vanishes, and the system is an open quantum harmonic oscillator \cite{EB_2005}. In this case, if the initial system state is Gaussian \cite{P_2010}, the conditional quantum state remains Gaussian
with the time-varying mean vector $\mu$ and the real part $\Sigma$ of the quantum covariance matrix of the system variables:
\begin{equation}
\label{muSigma}
    \mu := \pi_t(X),
    \qquad
    \Sigma:= \Re \pi_t((X-\mu)(X-\mu)^{\rT}),
\end{equation}
where  $\Sigma + i\Theta\succcurlyeq 0 $ in view of the Heisenberg uncertainty principle \cite{H_2001}.
The corresponding Gaussian QCF is given by
\begin{equation}
\label{Phigauss}
    \Phi_{\mu,\Sigma}(u) = \re^{i\mu^{\rT} u - \frac{1}{2}\|u\|_{\Sigma}^2},
    \qquad
    u \in \mR^n,
\end{equation}
where $\|u\|_{\Sigma}:= |\sqrt{\Sigma}u|$. The following theorem describes the time evolution of the parameters of the posterior Gaussian state.

%%%%%%%%%%%%%%%%%%%%%%%%%%%%%%%%%%%%%%%%%%%%%%%%%%%%%%%%%%%%%%%%%%%%%%%%%%%%%%%%%%%%%%%%%%%%%%%%%%%%
\begin{thm}
\label{th:linfilt}
Suppose the system dynamics are linear and specified by (\ref{N}) and (\ref{h0}) with $\Psi=0$, and the initial system state is Gaussian. % with mean vector $\mu_0$ and  quantum covariance matrix $\Sigma_0 + i\Theta$.
Then the parameters $\mu$ and $\Sigma$ of the posterior Gaussian state in (\ref{muSigma}) satisfy
\begin{align}
\label{dmu}
    \rd \mu
    & =
    A\mu \rd t + 2(P+Q\Sigma)^{\rT}K\rd \chi,\\
\label{Sigmadot}
    \dot{\Sigma}
    & =
    A\Sigma +\Sigma A^{\rT} + BB^{\rT}
     -
    4
    (P+Q\Sigma)^{\rT}KFF^{\rT}K^{\rT}(P+Q\Sigma).
\end{align}
Here, the matrices $A$ and $B$ are given by (\ref{AB}), the matrices  $P,Q\in \mR^{\frac{m}{2}\x n}$ are related by
\begin{equation}
\label{PQ}
    P:=
        {\begin{bmatrix}
            -E_1 & E_2
        \end{bmatrix}}
        N
        \Theta,
        \qquad
        Q:=
        {\begin{bmatrix}
            E_2 & E_1
        \end{bmatrix}}N
\end{equation}
to the matrices (\ref{E12}),
and the innovation process $\chi$ is driven by the measurements as
\begin{equation}
\label{dchilin}
    \rd \chi
    =
        \rd Z
        -
        2
        FF^{\rT}K^{\rT}
        Q \mu\rd t.
\end{equation}
\hfill$\square$
\end{thm}
%%%%%%%%%%%%%%%%%%%%%%%%%%%%%%%%%%%%%%%%%%%%%%%%%%%%%%%%%%%%%%%%%%%%%%%%%%%%%%%%%%%%%%%%%%%%%%%%%%%%
\begin{proof}
By substituting the Gaussian QCF from (\ref{Phigauss}) into (\ref{dPhihlin}) with $\Psi=0$, and using the identity $\d_u \Phi_{\mu,\Sigma}(u) = \Phi_{\mu,\Sigma}(u) (i\mu - \Sigma u)$, it follows that
\begin{align}
\nonumber
\frac{\rd \Phi_{\mu,\Sigma}}{\Phi_{\mu,\Sigma}}
& = 
\Big(u^{\rT}A (i\mu - \Sigma u) - \frac{1}{2}|B^{\rT}u|^2 \Big)\rd t
 + 2
        i
        u^{\rT}
        (
        P+Q\Sigma
        )^{\rT}K \rd \chi\\
\label{dPhilin}
 & = 
iu^{\rT}(A \mu\rd t + 2(P+Q\Sigma)^{\rT}K \rd \chi)
     - \frac{1}{2}u^{\rT}(A\Sigma + \Sigma A^{\rT} + BB^{\rT})u\rd t,
\end{align}
where $P$, $Q$ are the matrices from (\ref{PQ}). On the other hand,
application of the classical Ito lemma \cite{KS_1991} to $\ln\Phi_{\mu,\Sigma}$ (which depends on time only through $\mu$ and $\Sigma$) yields
\begin{align}
\nonumber
    \frac{\rd \Phi_{\mu,\Sigma}}{\Phi_{\mu,\Sigma}}
    & =
    \rd \ln\Phi_{\mu,\Sigma} + \frac{1}{2} \left(\frac{\rd \Phi_{\mu,\Sigma}}{\Phi_{\mu,\Sigma}}\right)^2\\
\label{dPhilin2}
     & = 
    iu^{\rT}\rd \mu
      - \frac{1}{2} u^{\rT}\big(\rd \Sigma +4 (P+Q\Sigma)^{\rT}K FF^{\rT}K^{\rT} (P+Q\Sigma)\rd t\big)u.
\end{align}
Here, use is made of  the relation $\rd \ln\Phi_{\mu,\Sigma}(u) = iu^{\rT}\rd \mu - \frac{1}{2}u^{\rT}\rd \Sigma u$ together with the quadratic variation $\left(\frac{\rd \Phi_{\mu,\Sigma}}{\Phi_{\mu,\Sigma}}\right)^2 = -4u^{\rT}(P+Q\Sigma)^{\rT}K FF^{\rT}K^{\rT} (P+Q\Sigma)u \rd t$ of the complex-valued diffusion process in (\ref{dPhilin}) and the diffusion matrix $FF^{\rT}$ of the innovation process $\chi$. The right-hand sides of (\ref{dPhilin}) and (\ref{dPhilin2}) are quadratic functions of $u \in \mR^n$. By matching the corresponding coefficients, it follows that $\mu$ satisfies the SDE (\ref{dmu}) while $\Sigma$ satisfies the ODE (\ref{Sigmadot}). Also, (\ref{dchilin}) follows from (\ref{dchihlin}) in view of (\ref{piX}), (\ref{muSigma}) and (\ref{PQ}).
\end{proof}
%%%%%%%%%%%%%%%%%%%%%%%%%%%%%%%%%%%%%%%%%%%%%%%%%%%%%%%%%%%%%%%%%%%%%%%%%%%%%%%%%%%%%%%%%%%%%%%%%%%%

The SDE (\ref{dmu}) and the ODE (\ref{Sigmadot}) are the quantum Kalman filter equations for the case of linear-Gaussian system dynamics. Similarly to the covariance dynamics of the usual Kalman filter \cite{AM_1979} for classical systems, (\ref{Sigmadot}) is organized as a differential Riccati equation (though with different matrices) which reduces to the Lyapunov  ODE $    \dot{\Sigma}
    =
    A\Sigma +\Sigma A^{\rT} + BB^{\rT}
$
in the absence of measurements. If $\Psi\ne 0$ in (\ref{h0}), then the QSDE (\ref{dXhlin}) is no longer linear, the integral operator term of  the SIDE (\ref{dPhihlin}) comes into effect, and the Gaussian QCFs (\ref{Phigauss}) can be used only as approximate solutions.

%%%%%%%%%%%%%%%%%%%%%%%%%%%%%%%%%%%%%%%%%%%%%%%%%%%%%%%%%%%%%%%%%%%%%%%%%%%%%%%%%%%%%%%%%%%%%%%%%%%
\section{A GAUSSIAN %MEAN-FIELD
APPROXIMATION OF THE POSTERIOR STATE}\label{sec:gauss}
%%%%%%%%%%%%%%%%%%%%%%%%%%%%%%%%%%%%%%%%%%%%%%%%%%%%%%%%%%%%%%%%%%%%%%%%%%%%%%%%%%%%%%%%%%%%%%%%%%%

For the class of nonlinear quantum stochastic systems with linear system-field coupling, described in the previous section,   we will now consider a Gaussian approximation of the actual posterior quantum state of the system using the criterion
 \begin{equation}
 \label{proxPhi}
    \|
        \Phi(t,\cdot)
        -
        \Phi_{\theta}
    \|^2
    =
    \int_{\mR^n}
    |\Phi(t,u) - \Phi_{\theta}(u)|^2
    \rd u
    \longrightarrow
    \min.
 \end{equation}
 Here, $\|\varphi\|$ denotes  the norm in the Hilbert space $L^2(\mR^n)$ of square integrable complex-valued functions on $\mR^n$, % with the inner product $\bra \varphi, \psi\ket:= \int_{\mR^n} \overline{\varphi(u)}\psi(u)\rd u$,
 and the minimization is over the parameter $\theta := (\mu,\Sigma) \in \mR^n\x \mS_n$ of the Gaussian QCF in (\ref{Phigauss}) subject to the constraint $\Sigma + i\Theta \succcurlyeq 0$, where $\mS_n$ denotes the subspace of real symmetric matrices of order $n$.
 %The set of admissible values of $\theta$ is denoted by
% \begin{equation}
% \label{dig}
%    \digamma
%    :=
%    \big\{
%        \theta
%        :=
%        (\mu,\Sigma)\in \mR^n\x \mS_n:\
%        \Sigma +i\Theta\succcurlyeq 0
%    \big\}.
%\end{equation}
The set $\mR^n\x \mS_n$ is a Hilbert space with the direct-sum inner product generated from the Euclidean inner product in  $\mR^n$ and the Frobenius inner product
 $\bra \Sigma_1,  \Sigma_2\ket:= \Tr (\Sigma_1  \Sigma_2)$ in $\mS_n$.
 The squared $L^2$-distance in (\ref{proxPhi}) is not the only possible proximity criterion. For example, \cite{V_2015c} employs    the second-order relative Renyi entropy \cite{R_1961}  in order to quantify the deviation of the actual QPDF $\mho$ from the Gaussian QPDFs
 \begin{equation}
\label{mhogauss}
    \mho_{\theta}(x) := \frac{(2\pi)^{-n/2}}{\sqrt{\det \Sigma}}\re^{-\frac{1}{2}\|x-\mu\|_{\Sigma^{-1}}^2},
    \qquad
    x\in \mR^n,
\end{equation}
provided $\Sigma\succ 0$ (the latter assumption also makes $\Phi_{\theta}$ square integrable). Unlike the relative entropy, (\ref{proxPhi}) treats the actual and approximating
QCFs equally and has the same form in terms of the QPDFs (\ref{mho}) and (\ref{mhogauss}) due to the Plancherel identity:
\begin{equation}
\label{proxmho}
    \|
        \mho(t,\cdot)
        -
        \mho_{\theta}
    \|^2
    =
    (2\pi)^{-n}
    \|
        \Phi(t,\cdot)
        -
        \Phi_{\theta}
    \|^2.
\end{equation}
If the actual posterior QCF $\Phi$ (or the corresponding posterior QPDF $\mho$) were known, then $\theta$ could be chosen,  at every moment of time, so as to minimize the cost in  (\ref{proxPhi}) (or equivalently, (\ref{proxmho})).  However, in the nonlinear case $\Psi\ne 0$, when the actual posterior QCF and  QPDF are difficult to find, the parameter $\theta$ can be evolved ``along'' the orthogonal projection (in the $L^2$ sense) of the Ito differential $\rd \Phi$ of the random field $\Phi$ from  (\ref{dPhihlin}) (whose right-hand side is computed at $\Phi = \Phi_{\theta}$) onto the tangent space of differentials $\rd \Phi_{\mu,\Sigma}$ of the Gaussian QCFs in (\ref{dPhilin2}). This approach (whose general idea is similar to that in \cite{VM_2005} and references therein) leads to a modified version of the quantum Kalman filter equations (\ref{dmu}) and (\ref{Sigmadot}):
\begin{align}
\label{dmuprox}
    \rd \mu
     =&
    (A\mu + \lambda)\rd t + 2(P+Q\Sigma)^{\rT}K\rd \chi,\\
\label{Sigmadotprox}
    \dot{\Sigma}
    =&
    A\Sigma +\Sigma A^{\rT} + BB^{\rT} + \sigma
     -
    4
    (P+Q\Sigma)^{\rT}KFF^{\rT}K^{\rT}(P+Q\Sigma),
\end{align}
where the additional terms $\lambda \in \mR^n$ and $\sigma \in \mS_n$ are found as a solution of the minimization problem
\begin{equation}
\label{lamsig}
    \|\fF_{\mu,\Sigma}(\lambda, \sigma,\cdot)\|^2
    =
    \int_{\mR^n}
    |\fF_{\mu,\Sigma}(\lambda, \sigma,u)|^2
    \rd u
    \longrightarrow \min.
\end{equation}
Here,
\begin{equation}
\label{fF}
    \fF_{\mu,\Sigma}(\lambda, \sigma,u)
    :=
    \fG_{\mu,\Sigma}(u)
    +
    \Big(iu^{\rT}\lambda - \frac{1}{2}u^{\rT}\sigma u\Big)
    \Phi_{\mu,\Sigma}(u),
\end{equation}
and
\begin{equation}
\label{fG}
    \fG_{\mu,\Sigma}(u)
    :=
        2
    \int_{\mR^d}
    \sin(u^{\rT}\Theta S^{\rT} v)
    \Psi(v)
    \Phi_{\mu,\Sigma}(u+S^{\rT}v)\rd v
\end{equation}
denotes the negative of the integral operator term in (\ref{dPhihlin}) which is contributed by the nonquadratic part of the system Hamiltonian. Both $\fF_{\mu,\Sigma}(\lambda, \sigma,u)$ and $\fG_{\mu,\Sigma}(u)$ are Hermitian with respect to $u\in \mR^n$.
The following theorem computes the correction terms $\lambda$ and $\sigma$.

%%%%%%%%%%%%%%%%%%%%%%%%%%%%%%%%%%%%%%%%%%%%%%%%%%%%%%%%%%%%%%%%%%%%%%%%%%%%%%%%%%%%%%%%%%%%%%%%%%%%
\begin{thm}
\label{th:gauss}
If $\Sigma\succ 0$, then the minimum in (\ref{lamsig}) is achieved at a unique point
\begin{align}
\label{lamopt0}
    \lambda
    =&
    \Re
    \int_{\mR^d}
    \re^{\frac{1}{4}z^{\rT} \Sigma^{-1}z}
    z\Big|_{z=-(\Sigma + i\Theta)S^{\rT}v}^{(-\Sigma + i\Theta)S^{\rT}v}
    \Phi_{\mu,\Sigma}(S^{\rT}v)
    \Psi(v)
    \rd v,\\
\label{sigopt0}
    \sigma
    =&
    2\cS_{\Sigma}^{-1}
    \Big(
    \Im
    \int_{\mR^d}
    \re^{\frac{1}{4}z^{\rT} \Sigma^{-1}z}
    \Big(\Sigma^{-1} +\frac{1}{2} \Sigma^{-1}zz^{\rT}\Sigma^{-1}\Big)\Big|_{z=-(\Sigma + i\Theta)S^{\rT}v}^{(-\Sigma + i\Theta)S^{\rT}v}
    \Phi_{\mu,\Sigma}(S^{\rT}v)
    \Psi(v)
    \rd v\Big),
\end{align}
where $\cS_{\Sigma}$ is a positive definite self-adjoint operator on $\mS_n$ given by
\begin{equation}
\label{cS}
    \cS_{\Sigma}(\sigma)
    :=
  \Sigma^{-1}\sigma \Sigma^{-1}
     +
    \frac{1}{2}
    \Bra \Sigma^{-1}, \sigma\Ket
    \Sigma^{-1}.
\end{equation}
\hfill$\square$
\end{thm}
%%%%%%%%%%%%%%%%%%%%%%%%%%%%%%%%%%%%%%%%%%%%%%%%%%%%%%%%%%%%%%%%%%%%%%%%%%%%%%%%%%%%%%%%%%%%%%%%%%%%
\begin{proof}
The function being minimized in (\ref{lamsig}) is a convex quadratic function of $(\lambda,\sigma)\in \mR^n\x \mS_n$, whose Frechet  differentiation   leads to the necessary  conditions of optimality:
\begin{align}
\label{lamopt}
    \Im
    \int_{\mR^n}
    \overline{\Phi_{\mu,\Sigma}(u)}\,
    \fF_{\mu,\Sigma}(\lambda, \sigma,u)
    u
    \rd u
    &=0,\\
\label{sigopt}
    \Re
    \int_{\mR^n}
    \overline{\Phi_{\mu,\Sigma}(u)}\,
    \fF_{\mu,\Sigma}(\lambda, \sigma,u)
        uu^{\rT}
    \rd u
    &=0,
\end{align}
where use is made of the relations $\d_{\lambda} \fF_{\mu,\Sigma}(\lambda, \sigma,u) = i\Phi_{\mu,\Sigma}(u)u$ and
$\d_{\sigma} \fF_{\mu,\Sigma}(\lambda, \sigma,u) = -\frac{1}{2}\Phi_{\mu,\Sigma}(u)uu^{\rT}$ which follow from (\ref{fF}). Note that
\begin{equation}
\label{gauss}
    \pi^{-n/2}
    \sqrt{\det\Sigma}\,
    |\Phi_{\mu,\Sigma}(u)|^2
    =
    \pi^{-n/2}
    \sqrt{\det\Sigma}\,
    \re^{-\|u\|_{\Sigma}^2},
\end{equation}
as a function of $u\in \mR^n$, is a Gaussian PDF with zero mean and covariance matrix $\frac{1}{2}\Sigma^{-1}$. Hence,
\begin{align}
\label{uu}
    \pi^{-n/2}\sqrt{\det\Sigma} \int_{\mR^n}
    |\Phi_{\mu,\Sigma}(u)|^2 uu^{\rT}\rd u
    = &
    \frac{1}{2}\Sigma^{-1},\\
\label{uuuu}
    \pi^{-n/2}\sqrt{\det\Sigma} \int_{\mR^n}
    |\Phi_{\mu,\Sigma}(u)|^2 uu^{\rT}\sigma uu^{\rT}\rd u
    = &
    \frac{1}{2}\cS_{\Sigma}(\sigma).
\end{align}
Here, the operator $\cS_{\Sigma}$ is given by (\ref{cS}) and originates from the relation $\bE(\xi\xi^{\rT}\sigma \xi\xi^{\rT}) = 2C\sigma C + \bra C,\sigma\ket C$ for any $\sigma \in \mS_n$ and an $\mR^n$-valued Gaussian random vector $\xi$ with zero mean and covariance matrix $C$, which follows from the Isserlis theorem \cite{I_1918,J_1997}. The mixed moments of arbitrary odd order for the entries of such a vector vanish.
Positive definiteness (and hence, invertibility) of the operator $\cS_{\Sigma}$ follows from the inequalities $\Bra \sigma, \cS_{\Sigma}(\sigma)\Ket \> \Tr((\Sigma^{-1/2}\sigma \Sigma^{-1/2})^2)>0$ for any $\sigma \in \mS_n \setminus \{0\}$. Now, a combination of (\ref{lamopt}) and (\ref{uu}) with (\ref{fG}) allows $\lambda$ to be uniquely found as
\begin{align}
\nonumber
    \lambda
    =&
    -2
    \pi^{-n/2}\sqrt{\det\Sigma} \,
    \Sigma
    \Im
    \int_{\mR^n}
    \overline{\Phi_{\mu,\Sigma}(u)}
    \fG_{\mu,\Sigma}(u)
    u
    \rd u\\
\nonumber
    =&
    -4
    \Sigma
    \Im
    \int_{\mR^d}
    \pi^{-n/2}\sqrt{\det\Sigma}
    \int_{\mR^n}
    \re^{-\|u\|_{\Sigma}^2}
    \sin(u^{\rT}\Theta S^{\rT} v)
    \frac{    \Phi_{\mu,\Sigma}(u+S^{\rT}v)}{\Phi_{\mu,\Sigma}(u)}
    u
    \rd u
        \Psi(v)
\rd v\\
\nonumber
    =&
    -4
    \Sigma
    \Im
    \int_{\mR^d}
    \pi^{-n/2}\sqrt{\det\Sigma}
    \int_{\mR^n}
    \re^{-\|u\|_{\Sigma}^2}
    \sin(u^{\rT}\Theta S^{\rT} v)
    \re^{-u^{\rT}\Sigma S^{\rT}v}
    u
    \rd u
    \Phi_{\mu,\Sigma}(S^{\rT}v)
        \Psi(v)
\rd v\\
\label{FFF}
    =&
    -4
    \Sigma
    \Im
    \int_{\mR^d}
    \frac{\phi_{\Sigma}'((-\Sigma + i\Theta) S^{\rT}v)-\phi_{\Sigma}'(-(\Sigma +i\Theta) S^{\rT}v)}{2i}
    \Phi_{\mu,\Sigma}(S^{\rT}v)
        \Psi(v)
\rd v,
\end{align}
where use is made of the identity
$
    \frac{    \Phi_{\mu,\Sigma}(u+w)}{\Phi_{\mu,\Sigma}(u)\Phi_{\mu,\Sigma}(w)}
    =
    \re^{-u^{\rT} \Sigma w}
$ for the Gaussian QCFs (\ref{Phigauss}) together with
the moment-generating function
\begin{equation}
\label{mom}
    \phi_{\Sigma}(z)
    :=
    \pi^{-n/2}\sqrt{\det\Sigma}
    \int_{\mR^n}
    \re^{-\|u\|_{\Sigma}^2+u^{\rT}z}
    \rd u =
    \re^{\frac{1}{4}z^{\rT}\Sigma^{-1}z},
    \quad
    z\in \mC^n,
\end{equation}
for the Gaussian PDF (\ref{gauss}).
Substitution of  the
gradient vector $\phi_{\Sigma}'(z) = \frac{1}{2}\phi_{\Sigma}(z)\Sigma^{-1}z$ into (\ref{FFF}) leads to (\ref{lamopt0}). By a similar reasoning, a combination of (\ref{sigopt}) and (\ref{uuuu}) with (\ref{fF}) and (\ref{fG}) allows $\sigma$ to be uniquely found as
\begin{align}
\nonumber
    \sigma
    & =
    4
    \cS_{\Sigma}^{-1}
    \Big(
    \pi^{-n/2}\sqrt{\det\Sigma}\,
    \Re
    \int_{\mR^n}
    \overline{\Phi_{\mu,\Sigma}(u)}
    \fG_{\mu,\Sigma}(u)
    uu^{\rT}
    \rd u
    \Big)\\
\nonumber
    & =
    8
\cS_{\Sigma}^{-1}
    \Big(
    \Re
    \int_{\mR^d}
    \pi^{-n/2}\sqrt{\det\Sigma}
    \int_{\mR^n}
    \re^{-\|u\|_{\Sigma}^2}
    \sin(u^{\rT}\Theta S^{\rT} v)
    \frac{    \Phi_{\mu,\Sigma}(u+S^{\rT}v)}{\Phi_{\mu,\Sigma}(u)}
    uu^{\rT}
    \rd u
        \Psi(v)
\rd v\Big)\\
\nonumber
    & =
    8
\cS_{\Sigma}^{-1}
    \Big(
    \Re
    \int_{\mR^d}
    \pi^{-n/2}\sqrt{\det\Sigma}
    \int_{\mR^n}
    \re^{-\|u\|_{\Sigma}^2}
    \sin(u^{\rT}\Theta S^{\rT} v)
    \re^{-u^{\rT}\Sigma S^{\rT}v}
    uu^{\rT}
    \rd u
    \Phi_{\mu,\Sigma}(S^{\rT}v)
        \Psi(v)
\rd v\Big)\\
\label{FFFF}
    & =
    8
    \cS_{\Sigma}^{-1}
    \Big(
    \Re
    \int_{\mR^d}
    \frac{\phi_{\Sigma}''((-\Sigma + i\Theta) S^{\rT}v)-\phi_{\Sigma}''(-(\Sigma +i\Theta) S^{\rT}v)}{2i}
    \Phi_{\mu,\Sigma}(S^{\rT}v)
        \Psi(v)
\rd v\Big).
\end{align}
Substitution of the Hessian matrix $\phi_{\Sigma}''(z) = \frac{1}{2}\phi_{\Sigma}(z)(\Sigma^{-1} +\frac{1}{2} \Sigma^{-1}zz^{\rT}\Sigma^{-1})$ of (\ref{mom}) into (\ref{FFFF}) leads to (\ref{sigopt0}).
\end{proof}
%%%%%%%%%%%%%%%%%%%%%%%%%%%%%%%%%%%%%%%%%%%%%%%%%%%%%%%%%%%%%%%%%%%%%%%%%%%%%%%%%%%%%%%%%%%%%%%%%%%%

The equations (\ref{lamopt0}) and (\ref{sigopt0}) provide integral representations of the correction terms $\lambda$ and $\sigma$ in the modified quantum Kalman filter (\ref{dmuprox}), (\ref{Sigmadotprox}) as nonlinear functions of $\mu$ and $\Sigma$. These integrals involve the spatial Fourier transform $\Psi$ of the nonquadratic part of the Hamiltonian.  Their closed-form  evaluation is possible, for example, if $\Psi$ is a linear combination of quadratic-exponential functions (see \cite[Section 9]{V_2015c}), which corresponds to the presence of Gaussian-shaped  ``bumps''   in the potential energy of the system \cite{FM_2004}. This consideration can be used in order to apply the  above results to open quantum systems with multiextremum energy landscapes. However, open questions  in regard to the Gaussian approximation,  described above, include its error analysis and the study of conditions when (\ref{Sigmadotprox}) produces a physically meaningful matrix $\Sigma$ satisfying the Heisenberg uncertainty  principle $\Sigma + i\Theta \succcurlyeq 0$.

%%%%%%%%%%%%%%%%%%%%%%%%%%%%%%%%%%%%%%%%%%%%%%%%%%%%%%%%%%%%%%%%%%%%%%%%%%%%%%%%%%%%%%%%%%%%%%%%%%%%%%%%
\section{CONCLUSION}\label{sec:conc}
%%%%%%%%%%%%%%%%%%%%%%%%%%%%%%%%%%%%%%%%%%%%%%%%%%%%%%%%%%%%%%%%%%%%%%%%%%%%%%%%%%%%%%%%%%%%%%%%%%%%%%%%

For a class of quantum stochastic systems, whose Hamiltonian and coupling operators are represented in the Weyl quantization form, we have obtained a nonlinear SIDE for the evolution of the posterior QCF conditioned on  multichannel nondemolition measurements. This equation is driven by a classical diffusion process of innovations associated with the measurements. We have also considered a more specific form of the SIDE for the case of linear system-field coupling and outlined a Gaussian approximation of the posterior state governed by modified quantum Kalman filter equations. These ideas are applicable to the development of suboptimal quantum filtering algorithms which employ more complicated (for example, multi-Gaussian) approximations of the posterior QCF and QPDF.
Furthermore, the results of this paper can be extended to more general system dynamics, field states and measurement settings (such as nonlinear coupling, coherent and Gaussian states and photon counting measurements), for some of which the BKSE was considered in \cite{EWP_2015,GK_2010,N_2014} without using the Weyl quantization of the Hamiltonian and coupling operators.


\begin{thebibliography}{99}{%\scriptsize
%==============================================================================
\bibitem{AM_1979}
B.D.O.Anderson, and J.B.Moore,
\emph{Optimal Filtering},
Prentice Hall, New York, 1979.
%==============================================================================
%\bibitem{B_2009}
%A.Bart\'{o}k-P\'{a}rtay,
%Gaussian approximation potential:
%an interatomic potential derived from
%first principles quantum mechanics,
%PhD Thesis,
%University of Cambridge, 2009 (arXiv:1003.2817 [cond-mat.mtrl-sci], 14 March 2010).
%==============================================================================
\bibitem{B_1983}
V.P.Belavkin, On the theory of controlling observable quantum
systems, \emph{Autom. Rem. Contr.}, vol. 44, no. 2, 1983, pp. 178--188.
%==============================================================================
\bibitem{B_1989}
V.P.Belavkin, A stochastic calculus of quantum input-output
processes and quantum nondemolition filtering,
\emph{Itogi Nauki i Tekhniki. Ser. Sovrem. Probl. Mat. Nov.
Dostizh.}, vol. 36, 1989, pp. 29--67.
%%==============================================================================
%\bibitem{B_2010}
%V.P.Belavkin, Noncommutative dynamics and generalized master equations, \emph{Math. Notes}, vol. 87, no. 5, 2010, pp. 636--653.
%==============================================================================

%==============================================================================
%\bibitem{BH_1998}
%D.S.Bernstein, and W.M.Haddad,
%LQG control with an
%$H^{\infty}$ performance bound: a Riccati equation approach,
%\emph{IEEE Trans.
%Automat. Contr.}, vol. 34, no. 3, 1989, pp. 293--305.

%%==============================================================================
%\bibitem{B_1986}
%W.M.Boothby,
%\emph{An Introduction to Differentiable Manifolds and Riemannian Geometry}, 2nd Ed.,
%Academic Press, London, 1986.

%%==============================================================================
\bibitem{BVJ_2007}
L.Bouten, R.Van Handel, M.R.James,
An introduction to quantum filtering,
\emph{SIAM J. Control Optim.}, vol. 46, no. 6, 2007, pp. 2199--2241.

%%==============================================================================
%\bibitem{B_1967}
%J.L.Brown, A Generalized form of Price's theorem and its converse,
%\emph{IEEE Trans. Inform. Theory}, vol. 13, no. 1, 1967, pp. 27--30.
%%==============================================================================
%\bibitem{CT_1991}%***
%T.M.Cover, and J.A.Thomas, \emph{Elements of Information Theory},
%Wiley, New York, 1991.
%%==============================================================================
\bibitem{CH_1971}
C.D.Cushen, and R.L.Hudson, A quantum-mechanical central limit theorem,
\emph{J. Appl. Prob.}, vol. 8, no. 3, 1971, pp. 454--469.
%%==============================================================================
%\bibitem{DD_1970}
%C.Doleans-Dade, Quelques applications de la formule de changement
%de variables pour les semimartingales, \emph{Z. Wahrscheinlichkeitstheorie
%verw.}, vol. 16, 1970, pp. 181--194.
%%==============================================================================
\bibitem{DDJW_2006}
C.D'Helon, A.C.Doherty, M.R.James, and S.D.Wilson,
Quantum risk-sensitive control,
45th IEEE CDC,
San Diego, CA, USA, December 13--15, 2006, pp. 3132--3137.
%%==============================================================================
\bibitem{EB_2005}
S.C.Edwards, and V.P.Belavkin,
Optimal quantum filtering and
quantum feedback control,
{\tt arXiv:quant-ph/0506018v2, August 1,  2005}.
%%==============================================================================
\bibitem{EWP_2015}
M.F.Emzir, M.J.Woolley, and I.R.Petersen,
Quantum filtering for multiple diffusive and Poissonian measurements,
\emph{J. Phys. A: Math. Theor.}, vol. 48, 385302.
%%==============================================================================
%\bibitem{E_1998}
%L.C.Evans,
%\emph{Partial Differential Equations},
%American Mathematical Society, Providence, 1998.
%%==============================================================================
\bibitem{F_1989}
G.B.Folland, \emph{Harmonic Analysis in Phase Space}, Princeton University Press, Princeton, 1989.
%%%==============================================================================
\bibitem{FM_2004}
P.A.Frantsuzov, and V.A.Mandelshtam,
Quantum statistical mechanics with Gaussians:
equilibrium properties of van der Waals clusters,
\emph{J. Chem. Phys.}, vol. 121, no. 19, 2004, pp. 9247--9256.
%%==============================================================================
%\bibitem{GHP_2009}
%P.Gibilisco, F.Hiai, and D.Petz,
%Quantum covariance, quantum Fisher information, and the uncertainty principle,
%\emph{IEEE Trans.
%Inform. Theory.}, vol. 55, no. 1, 2009, pp. 439--443.
%%==============================================================================

\bibitem{GZ_2004}
C.W.Gardiner, and P.Zoller,
\emph{Quantum Noise}.
Springer, Berlin, 2004.
%%==============================================================================
\bibitem{GS_2004}
I.I.Gikhman, and A.V.Skorokhod, \emph{The Theory of Stochastic Processes}, Springer, Berlin,
2004.
%%==============================================================================
\bibitem{GKS_1976}
V.Gorini, A.Kossakowski, E.C.G.Sudarshan, Completely positive dynamical semigroups of N-level systems,
\emph{J. Math. Phys.}, vol. 17, no. 5, 1976, pp. 821--825.
%%==============================================================================
%\bibitem{G_1999}
%J.E.Gough, Dissipative canonical flows in classical and quantum
%mechanics, \emph{J. Math. Phys.}, vol. 40, no. 6, 1999, pp. 2805--2815.
%%==============================================================================
\bibitem{GBS_2005}
J.Gough, V.P.Belavkin, and O.G.Smolyanov,
Hamilton-Jacobi-Bellman equations for
quantum optimal feedback control,
\emph{J. Opt. B: Quantum Semiclass. Opt.}, vol. 7, 2005, pp. S237--S244.
%%==============================================================================
%\bibitem{GJ_2009}
%J.Gough, and M.R.James,
%Quantum feedback networks: Hamiltonian
%formulation,
%\emph{Commun. Math. Phys.},  vol. 287, 2009, pp. 1109--1132.
%%==============================================================================
\bibitem{GK_2010}
J.E.Gough, and C.K\"{o}stler,
Quantum filtering in coherent states,
\emph{Commun. Stoch. Anal.},
vol. 4, no. 4,  2010, pp.  505--521.
%%==============================================================================

\bibitem{GRS_2014}
J.Gough, T.S.Ratiu, and O.G.Smolyanov,
Feynman, Wigner, and Hamiltonian structures describing the dynamics of open quantum systems,
\emph{Doklady Maths.}, vol. 89, no. 1, 2014, pp. 68--71.
%%==============================================================================
\bibitem{GRS_2015}
J.Gough, T.S.Ratiu, and O.G.Smolyanov,
Wigner measures and quantum control,
\emph{Doklady Maths.}, vol. 91, no. 2, 2015, pp. 199--203.
%%==============================================================================
%\bibitem{G_2015}
%J.Gough, Symplectic noise and the classical analog of the Lindblad generator.
%Does the regression hypothesis also fail in classical physics? \emph{J. Stat. Phys.},
%vol. 160, no. 6, 2015, pp. 1709--1720.
%%==============================================================================
%\bibitem{G_2006}
%M.de~Gosson,
%\emph{Symplectic Geometry and Quantum Mechanics},
%Birk\-h\"{a}user, Basel, 2006.
%%==============================================================================
%\bibitem{G_2008}%***
%R.M.Gray, \emph{Entropy and Information Theory}, Springer, New York,
%2008.
%%==============================================================================
%\bibitem{G_1977}
%H.W.Guggenheimer,
%\emph{Differential Geometry},
%Dover, New York, 1977.
%%==============================================================================
%\bibitem{H_2008}
%N.J.Higham,
%\emph{Functions of Matrices}, SIAM, 2008.
%--------------------------------------------------------------------

\bibitem{Hi_2010}
B.J.Hiley,
On the relationship between the Wigner-Moyal and Bohm approaches to quantum mechanics: a step to a more general theory?,
\emph{Foundat. Phys.}, vol. 40, no. 4, 2010, pp. 356--367.
%%==============================================================================
\bibitem{H_1991}
A.S.Holevo, Quantum stochastic calculus,
\emph{J. Math. Sci.}, vol. 56, no. 5, 1991, pp. 2609--2624.
%%==============================================================================

%\bibitem{H_1996}
%A.S.Holevo, Exponential formulae in quantum stochastic calculus,
%\emph{Proc. Roy. Soc. Edinburgh}, vol. 126A, 1994, pp. 375--389.
%--------------------------------------------------------------------
\bibitem{H_2001}
A.S.Holevo, \emph{Statistical Structure of Quantum Theory}, Springer, Berlin, 2001.
%%==============================================================================
%\bibitem{HJ_2007}
%R.A.Horn, and C.R.Johnson,
%\emph{Matrix Analysis},
%Cambridge
%University Press, New York, 2007.
%%==============================================================================
%\bibitem{H_1967}
%L.H\"{o}rmander,  Hypoelliptic second order differential equations, \emph{ Acta Math.},
%vol. 119, no. 1, 1967, pp. 147--171
%%==============================================================================

\bibitem{H_1974}
R.L.Hudson, When is the Wigner quasi-probability density non-negative?
\emph{Rep. Math. Phys.}, vol. 6, no. 2, 1974, pp. 249--252.
%%==============================================================================
\bibitem{H_2010}
R.L.Hudson,
Quantum Bochner theorems and incompatible observables,
\emph{Kybernetika}, vol. 46, no. 6, 2010, pp. 1061--1068.
%%==============================================================================
\bibitem{HP_1984}
R.L.Hudson, and K.R.Parthasarathy,
Quantum Ito's formula and stochastic evolutions.
\emph{Commun. Math. Phys.}, vol. 93,  1984, pp. 301--323.
%%==============================================================================
%\bibitem{HM_1994}
%U.Helmke, and J.B.Moore,
%\emph{Optimization and Dynamical Systems},
%Springer, London, 1994.
%%==============================================================================
\bibitem{I_1918}
L.Isserlis, On a formula for the product-moment coefficient of any order of a normal
frequency distribution in any number of variables, \emph{Biometrika}, vol. 12, 1918, pp.
134--139.
%%==============================================================================

%\bibitem{JK_1998}
%K.Jacobs, and P.L.Knight,
%Linear quantum trajectories: applications to continuous projection measurements,
%\emph{Phys. Rev. A}, vol.  57, no. 4, 1998, pp. 2301--2310.

%%==============================================================================
%\bibitem{J_2005}
%M.R.James, A quantum Langevin formulation of risk-sensitive optimal control,
%\emph{J. Opt. B}, vol. 7, 2005, pp. S198--S207.
%%==============================================================================
%\bibitem{JG_2010}
%M.R.James, and J.E.Gough,
%Quantum dissipative systems and feedback control design by interconnection,
% \emph{IEEE Trans. Autom. Contr.},  vol. 55, no. 8, pp. 1806--1821.





%%==============================================================================
%\bibitem{JNP_2008}
%M.R.James, H.I.Nurdin, and I.R.Petersen,
%$H^{\infty}$ control of
%linear quantum stochastic systems,
%\emph{IEEE Trans.
%Automat. Contr.}, vol. 53, no. 8, 2008, pp. 1787--1803.
%
%%==============================================================================
\bibitem{J_1997}
S.Janson, \emph{Gaussian Hilbert Spaces}, Cambridge University Press, Cambridge, 1997.
%%==============================================================================
\bibitem{KS_1991}
I.Karatzas, and S.E.Shreve,
\emph{Brownian Motion and Stochastic Calculus}, 2nd Ed.,
Springer, New York, 1991.
%%==============================================================================
%\bibitem{K_1976}
%T.Kato, \emph{Perturbation Theory for Linear Operators}, 2nd Ed., Springer, Berlin, 1976.
%%==============================================================================
\bibitem{K_1941}
A.N.Kolmogorov, Interpolation and extrapolation of stationary random sequences, \emph{Izv. Akad. Nauk SSSR, Ser. Mat.}, 1941, pp. 3--14.


%%==============================================================================
\bibitem{KS_2008}
J.Kupsch, and O.G.Smolyanov,
Exact master equations describing reduced dynamics of the Wigner function,
\emph{J. Math. Sci.}, vol. 150, no. 6, 2008, pp.  2598--2608.



%%==============================================================================
%\bibitem{KS_1972}
%H.Kwakernaak, and R.Sivan,
%\emph{Linear Optimal Control Systems},
%Wiley, New York, 1972.
%%==============================================================================
\bibitem{L_1976}
G.Lindblad, On the generators of quantum dynamical semigroups,
\emph{Comm. Math. Phys.}, vol. 48, 1976, pp. 119--130.
%%==============================================================================
\bibitem{LS_2001}
R.S.Liptser, and A.N.Shiryaev,
\emph{Statistics of Random Processes: Applications}, Springer, Berlin, 2001.
%%==============================================================================
%\bibitem{MWPY_2014}
%S.Ma, M.J.Woolley, I.R.Petersen, and N.Yamamoto,
%Preparation of pure Gaussian states via cascaded quantum systems,
%arXiv:1408.2290 [quant-ph], 11 Aug 2014.
%%==============================================================================

%\bibitem{M_1988}
%J.R.Magnus,
%\emph{Linear Structures},
%Oxford University Press, New York, 1988.
%%==============================================================================
%\bibitem{M_1997}
%P.Malliavin,
%\emph{Stochastic Analysis},
%Springer, Berlin, 1997.
%%==============================================================================
%\bibitem{Mar_1988}
%G.I.Marchuk, \emph{Splitting Methods}, Nauka, Moscow, 1988.
%
%%%==============================================================================
\bibitem{MD_2015}
K.-P.Marzlin, and S.Deering,
The Moyal equation for open quantum
systems, \emph{J. Phys. A: Math. Theor.}, vol.  48, 2015, pp. 205301(13).
%%==============================================================================
%\bibitem{M_1964}
%E.McMahon, An extension of Price's theorem, \emph{IEEE Trans. Inform. Theory}, vol. 10, no. 2, 1964, p. 168.
%%==============================================================================
\bibitem{M_1998}
E.Merzbacher,
\emph{Quantum Mechanics}, 3rd Ed.,
Wiley, New York, 1998.
%%==============================================================================
%\bibitem{MJ_2012}
%Z.Miao, and M.R.James,
%Quantum observer for linear quantum
%stochastic systems, Proc. 51st IEEE Conf. Decision Control, Maui,
%Hawaii, USA, December 10-13, 2012, pp. 1680--1684.
%%%==============================================================================
%\bibitem{M_1929}
%P.M.Morse,
%Diatomic molecules according to the wave mechanics. II. Vibrational levels,
%\emph{Phys. Rev.}, vol. 34, 1929, pp. 57--64.

%%==============================================================================
\bibitem{M_1949}
J. E. Moyal, Quantum mechanics as a statistical theory, \emph{Proc. Cam.
Phil. Soc.}, vol. 45, 1949, pp. 99--124.
%%==============================================================================
\bibitem{N_2014}
H.I.Nurdin, Quantum filtering for multiple input multiple output systems driven by arbitrary zero-mean jointly Gaussian input fields, \emph{Russ. J. Math. Phys.}, vol. 21, no. 3, pp. 386--398.
%%==============================================================================
%
%\bibitem{NJP_2009}
%H.I.Nurdin, M.R.James, and I.R.Petersen,
%Coherent quantum LQG
%control,
%\emph{Automatica}, vol.  45, 2009, pp. 1837--1846.
%%==============================================================================
%\bibitem{PAMGUJ_2014}
%Y.Pan, H.Amini, Z.Miao, J.Gough, V.Ugrinovskii, and M.R.James,
%Heisenberg picture approach to the stability of quantum
%Markov systems,
%\emph{J. Math. Phys.}, vol. 55, 2014, pp. 062701--1--16.
%
%%==============================================================================
%\bibitem{P_1965}
%A.Papoulis, Comments on `An extension of Price's theorem' by McMahon, E.L.,
%\emph{IEEE Trans. Inform. Theory},
%vol. 11, no. 1,  1965, p. 154.
%%==============================================================================

\bibitem{P_1992}
K.R.Parthasarathy,
\emph{An Introduction to Quantum Stochastic Calculus},
Birk\-h\"{a}user, Basel, 1992.
%%==============================================================================
\bibitem{P_2010}
K.R.Parthasarathy,
What is a Gaussian state?
\emph{Commun. Stoch. Anal.}, vol. 4, no. 2, 2010, pp. 143--160.

%%==============================================================================
\bibitem{PS_1972}
K.R.Parthasarathy, and K.Schmidt,
\emph{Positive Definite Kernels, Continuous Tensor Products, and Central Limit Theorems of Probability Theory},
Springer-Verlag, Berlin, 1972.

%%==============================================================================
%\bibitem{P_2010}
%I.R.Petersen,
%Quantum linear systems theory,
%Proc. 19th Int. Symp. Math. Theor. Networks Syst., Budapest, Hungary, July 5--9, 2010, pp.  2173--2184.
%%==============================================================================
%\bibitem{PUJ_2012}
% I.R.Petersen, V.Ugrinovskii, and M.R.James, Robust stability of uncertain linear quantum systems, \emph{Phil. Trans. Royal Soc. A}, vol. 370, no. 1979,  2012,  pp. 5354--5363.
%%==============================================================================
%\bibitem{PBGM_1962}
%L.S.Pontryagin, V.G.Boltyanskii,
%R.V.Gamkrelidze, and E.F. Mishchenko,
%\emph{The Mathematical Theory of Optimal Processes},
%Wiley, New York, 1962.
%==============================================================================
%\bibitem{P_1958}
%R.Price, A useful theorem for nonlinear devices having Gaussian inputs,
%\emph{IRE Trans. Inform. Theory}, vol. 4, no. 2, 1958, pp. 69--72.
%==============================================================================
%\bibitem{RS_1975}
%M.Reed, and B.Simon, \emph{Methods of Modern Mathematical Physics. II: Fourier Analysis, Self-adjointness},
%Academic Press, San Diego, 1975.
%==============================================================================
\bibitem{R_1961}%***
A.Renyi, On measures of entropy and information, Proc. 4th Berkeley
Sympos. Math. Statist. Prob., I, 1961, pp. 547--561.
%==============================================================================
%\bibitem{SIS_2004}
%A.Serafini, F.Illuminati, and S.De Siena,
%Symplectic invariants, entropic measures and correlation of Gaussian states,
%\emph{J. Phys. B: At. Mol. Opt. Phys.}, vol. 37, 2004, pp. L21--L28.
%==============================================================================
%

\bibitem{S_1994}
J.J.Sakurai,
\emph{Modern Quantum Mechanics},
 Addison-Wesley, Reading, Mass., 1994.
%==============================================================================



%\bibitem{SP_2009}
%A.J.Shaiju, and I.R.Petersen,
%On the physical realizability of
%general linear quantum stochastic differential equations with
%complex coefficients,
%Proc. Joint 48th IEEE Conf. Decision Control \&
%28th Chinese Control Conf.,
%Shanghai, P.R. China, December 16--18, 2009, pp. 1422--1427.

%==============================================================================
%\bibitem{SP_2012}
%A.J.Shaiju, and I.R.Petersen,
%A frequency domain condition for the physical
%realizability of linear quantum systems,
%\emph{IEEE Trans. Automat. Contr.}, vol. 57, no. 8, 2012, pp. 2033--2044.
%==============================================================================
\bibitem{SVP_2014}
A.Kh.Sichani, I.G.Vladimirov, and I.R.Petersen,
Robust mean square stability of open quantum stochastic systems
with Hamiltonian perturbations in a Weyl quantization form,
Australian Control Conference, 2014, Canberra, Australia, 17-18 November 2014, pp. 83--88.
%%==============================================================================
%
%\bibitem{S_2000}
%R.Simon,
%Peres-Horodecki separability criterion for continuous variable systems,
%\emph{Phys. Rev. Lett.},
%vol. 84, no. 12, 2000, pp. 2726--2729.

%
%%==============================================================================
%%
%\bibitem{SIG_1998}
%R.E.Skelton, T.Iwasaki, and K.M.Grigoriadis,
%\emph{A Unified Algebraic Approach to Linear Control Design},
%Taylor \& Francis, London, 1998.
%%%==============================================================================
%\bibitem{S_1968}
%G.Strang, On the construction and comparison of difference schemes, \emph{SIAM J. Numer. Anal.}, vol. 5, no. 3,
%1968, pp. 506--517.
%%%==============================================================================
%\bibitem{S_2008}
%D.W.Stroock,
%Partial differential equations for probabilists,
%Cambridge University Press, Cambridge, 2008.
%%%==============================================================================
%\bibitem{SW_1997}
%H.J.Sussmann, and J.C.Willems,
%300 years of optimal control: from the brachystochrone to the maximum principle,
%\emph{Control Systems}, vol. 17, no. 3, 1997, pp. 32--44.

%\bibitem{V_1999}
%A.van den Bos, Nonlinear statistical signal processing: useful theorems and their application, Proc. IEEE-EURASIP Workshop on Nonlinear Signal and Image Processing (NSIP'99), Antalya, Turkey, June 20--23, 1999, pp. 603--606.
%{\tt http://www.eurasip.org/Proceedings/Ext/NSIP99/Nsip99/ papers/130.pdf}.
%%==============================================================================
\bibitem{VM_2005}
R.Van Handel, and H.Mabuchi, Quantum projection filter for a highly nonlinear
model in cavity QED, \emph{J. Opt. B: Quantum Semiclass. Opt.}, vol. 7, 2005, pp. S226--S236.
%%==============================================================================
%\bibitem{V_1984}
%V.S.Varadarajan, \emph{Lie Groups, Lie Algebras, and Their Representations},
%Springer-Verlag, New York, 1984.

%%==============================================================================
%\bibitem{VFGE_2012}
%V.Veitch, C.Ferrie, D.Gross, and J.Emerson,
%Negative quasi-probability as a resource for quantum computation,
%\emph{New J. Phys.}, vol. 14, 2012, pp. 113011(1)--113011(21).
%
%%==============================================================================
%\bibitem{V_1971}
%V.S.Vladimirov,
%\emph{Equations of Mathematical Physics},
%M.Dekker,
%New York, 1971.
%%%==============================================================================
\bibitem{V_2002}
V.S.Vladimirov,
\emph{Methods of the Theory of Generalized Functions},
Taylor \& Francis, London, 2002.
%==============================================================================
%\bibitem{VP_2010a}
%I.G.Vladimirov, and I.R.Petersen,
%Minimum relative entropy state transitions in linear stochastic
%systems: the continuous time case,
%Proc. 19th Int. Symp. Math. Theor. Networks Syst., Budapest, Hungary, July 5--9,  2010, pp.  51--58.
%%%==============================================================================
%\bibitem{VP_2010b}
%I.G.Vladimirov, and I.R.Petersen,
%Hardy-Schatten norms of systems, output energy cumulants and linear quadro-quartic  Gaussian control,
%Proc. 19th Int. Symp. Math. Theor. Networks Syst., Budapest, Hungary, July 5--9,  2010, pp.  2383--2390.
%%==============================================================================
%\bibitem{VP_2011a}
%I.G.Vladimirov, and I.R.Petersen,
%A quasi-separation principle and Newton-like scheme for coherent quantum LQG control, 	
%18th IFAC World Congress, Milan, Italy, 28 August--2 September, 2011, pp. 4721--4727.
%(preprint:  arXiv:1010.3125v2 [quant-ph], 15 April, 2011).
%%==============================================================================
%\bibitem{VP_2011b}
%I.G.Vladimirov, and I.R.Petersen,
%A dynamic programming approach to finite-horizon coherent quantum LQG control, 	
%Proc. Australian Control Conference, Melbourne, 10--11 November, 2011, pp. 357--362.
%(preprint:  	arXiv:1105.1574v1 [quant-ph], 9 May, 2011).
%%==============================================================================
%\bibitem{VP_2012a}%***
%I.G.Vladimirov, and I.R.Petersen,
%Gaussian stochastic  linearization for open quantum systems
%using quadratic approximation of Hamiltonians, 	MTNS 2012, Melbourne, Victoria, 9-13 July 2012,
%%{\tt https://fwn06.housing.rug.nl/mtns/?page\_id=13},
%{\tt arXiv:1202.0946v1 [quant-ph], 5 February 2012}.
%%==============================================================================
%\bibitem{VP_2012b}
%I.G.Vladimirov, and I.R.Petersen,
%Risk-sensitive dissipativity of linear quantum stochastic systems
%under Lur'e type perturbations of Hamiltonians, 	Proc. Australian Control Conference, Sydney, Australia, 15-16 November 2012, pp. 247--252. %(preprint:  	arXiv:1205.3566v1 [quant-ph], 16 May 2012).
%%==============================================================================
%\bibitem{VP_2012c}
%I.G.Vladimirov, and I.R.Petersen,
%Characterization and moment stability analysis of quasilinear quantum stochastic systems with quadratic coupling to external fields, Proc. 51st Conference on Decision and Control, IEEE, Maui, Hawaii, USA, 10-13 December  2012, pp. 1691--1696.
%%==============================================================================
%\bibitem{VP_2013a}
%I.G.Vladimirov, and I.R.Petersen,
%A quasi-separation principle and Newton-like scheme for coherent quantum LQG control,
%\emph{Syst. Contr. Lett.}, vol. 62, no. 7, 2013, pp. 550--559.
%%==============================================================================
%\bibitem{VP_2013b}
%I.G.Vladimirov, and I.R.Petersen,
%Coherent quantum filtering for physically realizable linear quantum plants,
%Proc. European Control Conference, IEEE, Zurich, Switzerland, 17-19 July 2013,  pp. 2717--2723.
%%==============================================================================
\bibitem{V_2015a}
I.G.Vladimirov,
A transverse Hamiltonian  variational technique for open quantum stochastic systems and its application to coherent quantum control, IEEE Multi-Conference on Systems and Control, 21-23 September 2015, Sydney, Australia, pp. 29--34.
%(arXiv:1506.04737v2 [quant-ph], 7 August 2015).
%%==============================================================================
\bibitem{V_2015b}
I.G.Vladimirov,
Weyl variations and local sufficiency of linear observers in the mean square optimal coherent quantum filtering problem,
Australian Control Conference, 5-6 November 2015, Gold Coast, Australia, pp. 93--98.
 %(arXiv:1506.07653 [quant-ph], 25 June 2015).
%%==============================================================================
\bibitem{V_2015c}
I.G.Vladimirov,
Evolution of quasi-characteristic functions in quantum stochastic systems with Weyl quantization
of energy operators, {\tt arXiv:1512.08751 [math-ph], 29 December 2015}.


%%==============================================================================
%\bibitem{W_1936}
%J.Williamson,
%On the algebraic problem concerning the normal forms of linear dynamical systems,
%\emph{Am. J. Math.}, vol. 58, no. 1, 1936, pp. 141--163.
%%==============================================================================
%\bibitem{W_1937}
%J.Williamson,
%On the normal forms of linear canonical transformations in dynamics,
%\emph{Am. J. Math.}, vol. 59, no. 3, 1937, pp. 599--617.
%%==============================================================================
%\bibitem{V_1968}
%H.C.Volkin, Iterated commutators and functions of operators, NASA Technical Note, D-4857, National Aeronautics and Space Administration, Washington, D.C., 1968, pp. 1--19.
%%==============================================================================
\bibitem{W_1949}
N.Wiener, \emph{Extrapolation, Interpolation, and Smoothing of Stationary Time Series}, Wiley, New York, 1949.
%%==============================================================================

%\bibitem{W_1967}
%R.M.Wilcox,
%Exponential operators and parameter differentiation in quantum physics,
%\emph{J. Math. Phys.}, vol.  8, no. 4, 1967, pp. 962--982.
%%==============================================================================
%\bibitem{W_1972}
%J.C.Willems, Dissipative dynamical systems. Part I: general
%theory, Part II: linear
%systems with quadratic supply rates, \emph{Arch. Rational Mech. Anal.}, vol. 45, no. 5, 1972, pp.
%321--351, 352--393.
%%--------------------------------------------------------------------
%\bibitem{W_1972II}
%J.C.Willems, Dissipative dynamical systems. Part II: linear
%systems with quadratic supply rates, \emph{Arch. Rational Mech.
%Anal.}, vol. 45, no. 5, 1972, pp. 352--393.
%--------------------------------------------------------------------
\bibitem{WM_2010}
H.M.Wiseman, and G.J.Milburn,
\emph{Quantum measurement and control},
Cambridge University Press,
Cambridge.
%--------------------------------------------------------------------
%\bibitem{YB_2009}
%N.Yamamoto, and L.Bouten,
%Quantum risk-sensitive estimation and robustness,
%\emph{IEEE Trans. Automat. Contr.}, vol. 54, no. 1, 2009, pp. 92--107.
%%==============================================================================
%\bibitem{Y_2009}
%M.Yanagisawa, Non-Gaussian state generation from linear elements
%via feedback, \emph{Phys. Rev. Lett.}, vol. 103, no. 20, 2009, pp. 203601--1--4.
%%==============================================================================

\bibitem{Y_1980}
K.Yosida, \emph{Functional Analysis}, 6th Ed., Springer, Berlin, 1980.
%%==============================================================================
%\bibitem{ZJ_2011}
%G.Zhang, and M.R.James, On the response of linear quantum stochastic
%systems to single-photon inputs and pulse shaping of photon
%wave packets, Proc. Australian Control Conference, Melbourne, 10-11
%November 2011, pp. 62--67.
}
\end{thebibliography}
\end{document}